%% file: paper.tex
\documentclass{article}

\usepackage{microtype}
\usepackage{graphicx}
\usepackage{subfigure}
\usepackage{booktabs} 

\usepackage{hyperref}

\usepackage[accepted]{icml2023}

\usepackage{amsmath}
\usepackage{amssymb}
\usepackage{mathtools}
\usepackage{amsthm}

\usepackage[acronym]{glossaries}
\glsdisablehyper
\usepackage{bm}
\usepackage{physics}
\usepackage{import}
\usepackage{graphicx}
\usepackage{pgfplots}

\pgfplotsset{compat=1.17}

\newcommand{\highlight}[1]{\boxed{#1}}
\usepackage{framed, color}
\definecolor{darkgreen}{rgb}{0.0, 0.5, 0.0}
\definecolor{orange}{rgb}{1.0, 0.49, 0.0}
\definecolor{modified}{rgb}{0.0, 0.0, 1.0}
\definecolor{cr}{rgb}{0.0, 0.0, 1.0}

\usepackage[capitalize,noabbrev]{cleveref}

\theoremstyle{plain}
\newtheorem{theorem}{Theorem}[section]

\newtheorem{lemma}[theorem]{Lemma}
\newtheorem{no-lemma}[theorem]{Lemma}
\newtheorem{corollary}[theorem]{Corollary}
\theoremstyle{definition}
\newtheorem{definition}[theorem]{Definition}

\theoremstyle{remark}

\usepackage[textsize=tiny]{todonotes}

\usepackage{stmaryrd}

\icmltitlerunning{Quantum Policy Gradient Algorithm with Optimized Action Decoding}

\begin{document}

\input{glossary.tex}

\twocolumn[
\icmltitle{Quantum Policy Gradient Algorithm with Optimized Action Decoding}




\begin{icmlauthorlist}
\icmlauthor{Nico Meyer}{iis,phys}
\icmlauthor{Daniel D. Scherer}{iis}
\icmlauthor{Axel Plinge}{iis}
\icmlauthor{Christopher Mutschler}{iis}
\icmlauthor{Michael J. Hartmann}{phys}
\end{icmlauthorlist}

\icmlaffiliation{iis}{Fraunhofer IIS, Fraunhofer Institute for Integrated Circuits IIS, Nuremberg, Germany}
\icmlaffiliation{phys}{Department of Physics, Friedrich-Alexander University Erlangen-Nuremberg (FAU), Erlangen, Germany}

\icmlcorrespondingauthor{Nico Meyer}{nico.meyer@iis.fraunhofer.de}

\icmlkeywords{Quantum Computing, Reinforcement Learning, Quantum Machine Learning, Quantum Reinforcement Learning, Policy Gradients, CartPole, FrozenLake, ContextualBandits, Variational Quantum Circuits}

\vskip 0.3in
]



\printAffiliationsAndNotice{}  

\begin{abstract}
Quantum machine learning implemented by \glspl{vqc} is considered a promising concept for the \glsentrylong{nisq} computing era. Focusing on applications in \glsentrylong{qrl}, we propose an action decoding procedure for a \glsentrylong{qpg} approach. We introduce a quality measure that enables us to optimize the classical post-processing required for action selection, inspired by local and global quantum measurements. The resulting algorithm demonstrates a significant performance improvement in several benchmark environments. With this technique, we successfully execute a full training routine on a $5$-qubit hardware device. Our method introduces only negligible classical overhead and has the potential to improve \gls{vqc}-based algorithms beyond the field of \glsentrylong{qrl}.
\end{abstract}

\input{main}


\bibliography{paper}
\bibliographystyle{icml2023}

\newpage
\appendix

\onecolumn
\input{appendix}



\end{document}

%% file: glossary.tex
\newacronym{vqc}{VQC}{variational quantum circuit}
\newacronym{rl}{RL}{reinforcement learning}
\newacronym{qrl}{QRL}{quantum reinforcement learning}
\newacronym{ml}{ML}{machine learning}
\newacronym{qml}{QML}{quantum machine learning}
\newacronym{dnn}{DNN}{deep neural network}
\newacronym{qc}{QC}{quantum computing}
\newacronym{nisq}{NISQ}{noisy intermediate-scale quantum}
\newacronym{pg}{PG}{policy gradient}
\newacronym{qpg}{QPG}{quantum policy gradient}
\newacronym{mdp}{MDP}{Markov Decision Process}
\newacronym{spsa}{SPSA}{simultaneous perturbation stochastic approximations}
\newacronym{fim}{FIM}{Fisher information matrix}
\newacronym{pdf}{PDF}{probability density function}

%% file: main.tex
\section{Introduction}

\Gls{rl} currently receives increasing attention due to its potential in a multitude of applications. In an \gls{rl} setup, an agent aims to learn a control strategy, i.e., a policy, for a specific problem. Training such a policy can require approximating a complex, multimodal distribution, which is often done with a \gls{dnn}. With increasing problem difficulty, this approach potentially has an undesirable sampling and model complexity \cite{Kakade_2003,Nielsen_2015,Poggio_2020}. Training data is obtained by interaction with the environment via actions, which returns a reward value and a new state. One can optimize the parameters of the policy to maximize the long-term reward with gradient-based techniques, forming a \gls{pg} algorithm. Real-world applications can be found e.g.\ in the domains of self-driving cars~\cite{Bojarski_2016} or MIMO beamforming~\cite{Maksymyuk_2018}.

Exploring the possibilities of other computing paradigms might elevate the impact of \gls{rl}, e.g.\ by circumventing the problems caused by increasing parameter complexity of \gls{dnn}-based models. \Gls{qc}, based on the idea of exploiting quantum mechanical properties for computation, might offer advantages in the approximation and sampling from complex distributions. Although the development of quantum computers is still in its infancy, a number of studies have already claimed experimental results beyond classical capabilities on specific tasks~\cite{Arute_2019,Zhong_2020,Wu_2021}.

\begin{figure}[t!]
\centering
\def\svgwidth{\columnwidth}
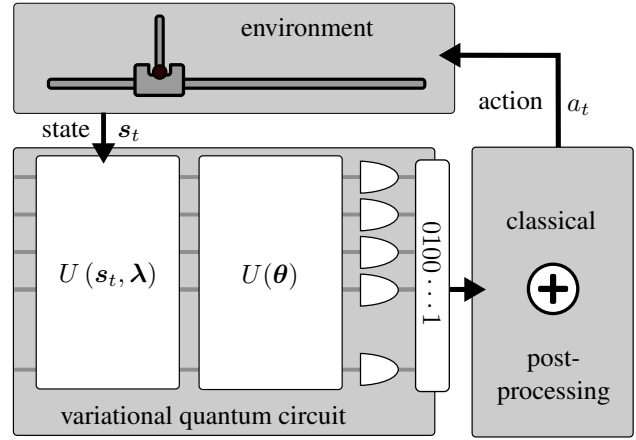
\caption{Proposed method: The prepared quantum state is measured in the computational basis. The results are post-processed (with a function maximizing our proposed \emph{globality measure}) before selecting an action. The update of the parameters $\Theta = (\bm{\theta},\bm{\lambda})$ utilizes the same post-processing scheme.}
\label{fig:overview}
\vspace*{-3mm}
\end{figure}

 The nowadays available quantum devices are considered \gls{nisq} hardware~\cite{Preskill_2018}, i.e., they only provide a limited amount of qubits that are heavily affected by noise. Therefore, a major part of current research focuses on strategies and algorithms that are able to cope with these restrictions, while at the same time aiming for computational power beyond what is possible classically. A promising idea suggests using \emph{variational quantum algorithms} as a platform for \gls{qml}~\cite{Benedetti_2019,Cerezo_2021}, for which a certain degree of resilience to the inevitable hardware noise has been reported~\cite{Li_2017,Moll_2018,Sharma_2020,Fontana_2021}. Variational quantum algorithms use a \emph{\gls{vqc}}, which incorporates trainable parameters, and a classical optimization routine to optimize these parameters. When viewing the \glspl{vqc} as function approximators, the property of universal function approximation holds under certain conditions~\cite{Goto_2021,Schuld_2021a}. For specific problems, variational quantum algorithms are known to exhibit provable quantum advantage~\cite{Liu_2021,Sweke_2021}.

\textbf{Contribution.} In variational quantum algorithms, it is necessary to extract classical information from the prepared quantum state. Whereas there has been work on obtaining a maximal amount of information about the state via a limited amount of measurements \cite{Huang_2020}, our goal is to group measured bitstrings from a readout in the computational basis, such that a well performing \gls{rl} strategy emerges. Those aspects have, to our knowledge, not yet been explored for \gls{vqc}-based \gls{qpg} algorithms. For the \gls{qrl} setup, we refer to this task as \emph{action decoding}. Motivated by the \texttt{RAW}-\gls{vqc} policy~\cite{Jerbi_2021} in \cref{sec:qpg}, we start with a formulation in terms of projective measurements in \cref{subsec:partitioning}. This is then decomposed into a measurement in the computational basis and the successive application of a classical post-processing function in \cref{subsec:globality_measure}. Our developed \emph{globality measure} allows to compare specific instances of those functions. Furthermore, we propose a routine to construct an optimal (w.r.t.\ the globality measure) post-processing function. It is important to mention that, in contrast to the approach by \citeauthor{Jerbi_2021}, our procedure is feasible for problems with large action spaces.

We observe a strong correlation between \gls{rl} performance and our globality measure in \cref{subsec:exp_results} in the \gls{rl} environments \texttt{CartPole}, \texttt{FrozenLake}, and different configurations of \texttt{ContextualBandits}. Training converges much faster (or even at all) for policies with an underlying post-processing function that has a large globality value. The results are supported by an analysis of the effective dimension and Fisher information spectrum in \cref{subsec:exp_fim}. As our technique only marginally increases the classical overhead (while reducing the required \gls{vqc} size) it suggests itself as a tool for execution on \gls{nisq} devices. To demonstrate the efficiency of our algorithm, we execute the full \gls{rl} training routine for a \texttt{ContextualBandits} problem on a $5$-qubit quantum hardware device in \cref{subsec:quantum_hardware}.


\section{Related Work}

A summary of the current body of work on \gls{qrl} can be found in Meyer et al.~\cite{Meyer_2022}. Specific \gls{qrl} routines have already been realized experimentally~\cite{Saggio_2021}. An early instance of \gls{vqc}-based \gls{qrl} proposes to use a \gls{vqc} as an approximator for the action-value function~\cite{Chen_2020}. We follow the lines of \citeauthor{Jerbi_2021}, which uses the \gls{vqc} for policy approximation, forming a \gls{qpg} algorithm. Additional work on the \gls{qpg} approach include an extension to quantum environments~\cite{Sequeira_2022}, and a modified parameter update to reduce sampling complexity~\cite{Meyer_2023}. The ideas of value-function and policy approximation are combined into actor-critic approaches~\cite{Wu_2020,Kwak_2021}, which also can benefit from our contribution.

While the algorithmic routine of \gls{qpg} follows the idea of classical \gls{pg}, the design of the \gls{vqc} function approximator is an ongoing research field. The typical architecture features three different blocks, i.e., a data encoding layer, (potentially multiple) variational layers with trainable parameters, and some measurement observables that extract information from the prepared quantum state. There are some guidelines for designing data encoding~\cite{Perez_2020,Schuld_2021a,Periyasamy_2022} and variational layers~\cite{Sim_2019,Kandala_2017}, based on the specific problem type. We focus on the necessary measurements, which require special attention in the context of quantum information theory~\cite{Braginsky_1995,Nielsen_2010}, and also \gls{qml} with \glspl{vqc}~\cite{Schuld_2018,Cerezo_2021,Schuld_2021b}. However, the question how to best measure \gls{vqc} outputs and classically post-process them to optimize \gls{qrl} performance, yet alone \gls{qpg} performance, is still open.



\begin{figure*}[t]
    \subfigure[Variational block with $1$-qubit $R_z$ and $R_y$ gates parameterized by $\bm{\theta}$.]{
       \includegraphics[width=0.26\linewidth]{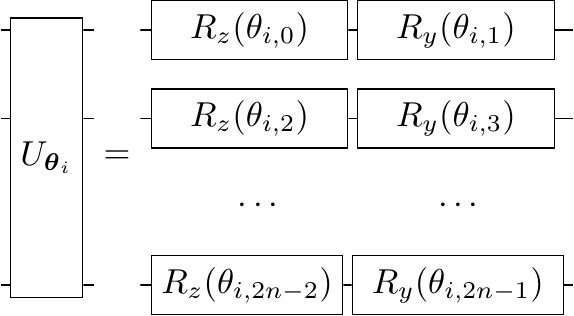}
    }
    \quad
    \subfigure[All-to-all structure of $CZ$ gates (hardware experiment uses $CX$ gates).]{
       \includegraphics[width=0.29\linewidth]{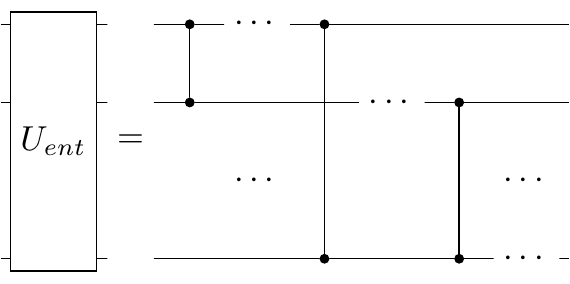}
    }
    \quad
    \subfigure[The $n$-dim \gls{rl} state $\bm{s}$ is encoded using $1$-qubit $R_y$ and $R_z$ rotations, with scaling parameters $\bm{\lambda}$.]{
       \includegraphics[width=0.36\linewidth]{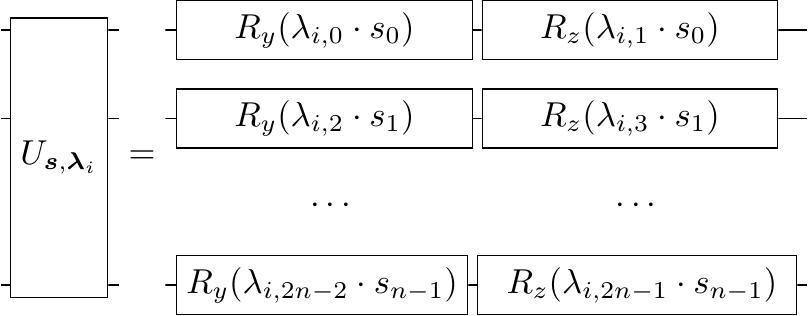}
    } \\
    \subfigure{
        \includegraphics[width=0.99\linewidth]{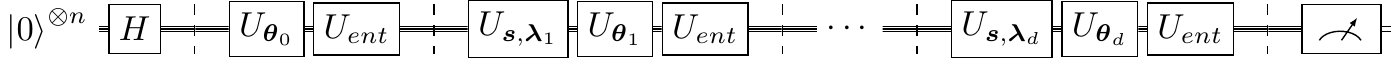}
    }
    \caption{\label{fig:vqc}Hardware-efficient quantum circuit, adapted from \citeauthor{Jerbi_2021}. The parameters are summarized in $\Theta = (\bm{\theta},\bm{\lambda})$. Depending on circuit depth $d$, the encoding blocks $U_{\bm{s},\bm{\lambda}}$, combined with variational blocks $U_{\bm{\theta}}$ and entanglement blocks $U_{ent}$, are repeated (i.e.\ for $d \geq 2$ data re-uploading~\cite{Perez_2020} is used). Measurements are performed in the computational basis.}
\end{figure*}


\section{\label{sec:qpg}Quantum Policy Gradient Algorithm}

\gls{rl} is an algorithmic concept to solve a complex task, where data is generated by interaction of an agent with an environment. The setup is usually described as a \gls{mdp}, i.e.\ a $5$-tuple $\left( \mathcal{S}, \mathcal{A}, \mathcal{R}, p, \gamma \right)$, where $\mathcal{S}$ is the state set, $\mathcal{A}$ is the set of available actions, $\mathcal{R} \subset \mathbb{R}$ is the reward space, $p:\mathcal{S} \times \mathcal{R} \times \mathcal{S} \times \mathcal{A} \to \left[ 0,1 \right]$ describes the environment dynamics, and $0 \leq \gamma \leq 1$ is a discount factor. At each timestep $t$, the agent observes the environment state $\bm{s}_t$, and decides on an action $a_t$. This decision is sampled from the current policy $\pi: \mathcal{S} \times  \mathcal{A} \to \left[ 0,1 \right]$, which defines a \gls{pdf} over all possible actions $a$ for a given state $\bm{s}$. The selected action is executed, and the agent receives a scalar reward $r_t \in \mathcal{R}$, after which the environment transitions to state $s_{t+1}$, following its dynamics $p$. The objective is to learn a policy, which maximizes the (discounted) return $G_t := \sum_{t'=t}^{H-1} \gamma^{t'-t} \cdot r_{t'}$ for some horizon $H<\infty$~\cite{Sutton_2018}.

Our work follows the hybrid \gls{qpg} algorithm proposed by \citeauthor{Jerbi_2021}. The approach is inspired by the classical \texttt{REINFORCE} idea with function approximation~\cite{Sutton_1999}, also referred to as \emph{vanilla policy gradient}. Its centerpiece is the parameterized policy $\pi_{\Theta}: \mathcal{S} \times \mathcal{A} \mapsto \left[ 0,1 \right]$, where $\Theta$ denotes the trainable parameters of the function approximator. The parameters are updated with a gradient ascent technique, i.e., $\Theta \gets \Theta + \alpha \cdot \nabla_{\Theta} J(\Theta)$, with learning rate $\alpha$, and some scalar performance measure $J(\Theta)$. The \emph{policy gradient theorem}~\cite{Sutton_1999} states the gradient of the performance measure as
\begin{equation}
    \label{eq:policy_gradient}
    \nabla_{\Theta} J(\Theta) = \mathbb{E}_{\pi_{\Theta}} \left[ \sum_{t=0}^{H-1} \nabla_{\Theta} \ln \pi_{\Theta} (a_t \mid \bm{s}_t) \cdot G_t \right].
\end{equation}
\noindent In practice, the expectation value in \cref{eq:policy_gradient} is approximated by averaging over several trajectories $\tau$ (i.e., sequences of current state, executed action, and received reward for some timesteps), that are generated by following the current policy $\pi_{\Theta}$. For a \gls{dnn}, the gradient of the log-policy w.r.t.\ the parameters can be obtained using backpropagation~\cite{Rumelhart_1986}. For \glspl{vqc} executed on quantum hardware one typically resorts to the parameter-shift rule~\cite{Mitarai_2018,Schuld_2019} -- as in this paper -- or simultaneous perturbation stochastic approximations (SPSA)~\cite{Spall_1998,Wiedmann_2023}.

\subsection{\label{subsec:architecture}VQC-Model Architecture}

We use a \gls{vqc} with a subsequent measurement as a replacement for the \gls{dnn}, which usually approximates the policy in deep \gls{rl}~\cite{Jerbi_2021}. The \gls{vqc} acts on $\ket{0}^{\otimes n}$ (where $\ket{0}$ denotes the $1$-qubit computational zero state) with the unitary $U_{\bm{s},\bm{\lambda},\bm{\theta}}$, which prepares the quantum state $\ket{\psi_{\bm{s},\bm{\lambda},\bm{\theta}}}$. We introduce only the fundamentals of \gls{qc} and refer the interested reader to \citeauthor{Nielsen_2010} for more details.

Similar to \citeauthor{Jerbi_2021} we use the hardware-efficient ansatz in \cref{fig:vqc}. Besides the variational parameters $\bm{\theta}$, there are also trainable scaling parameters $\bm{\lambda}$ to enhance the expressivity of the model. For ease of notation we denote $\Theta = (\bm{\theta}, \bm{\lambda})$.

Extracting classical information from the quantum state $\ket{\psi_{\bm{s},\Theta}}$ is crucial when using \glspl{vqc} in a hybrid algorithm. Usually, we measure some Hermitian operator $O$ to estimate the expectation $\expval{O}_{\bm{s},\Theta}$.
For a projective measurement and spectral decomposition $O = \sum_{i} \mu_i \ket{v_i} \bra{v_i}$, the post-measurement state corresponds to one of the eigenstates $\ket{v_i}$, and we observe the corresponding eigenvalue $\mu_i$ with probability $\expval{\psi_{\bm{s},\Theta} \mid v_i} \expval{v_i \mid \psi_{\bm{s},\Theta}}$.
We restrict our considerations to measuring only $1$-qubit Pauli observables $\sigma \in \left\{ X, Y, Z \right\}$.

\subsection{\label{subsec:policies}Reformulation of the \texttt{RAW}-VQC Policy}

\citeauthor{Jerbi_2021} define the \texttt{RAW}-\gls{vqc} and the \texttt{SOFTMAX}-\gls{vqc} and suggest, that the latter formulation is superior in terms of \gls{rl} performance. However, we argue in \cref{sec:restricted_softmax}, that is has several drawbacks w.r.t.\ circuit sampling complexity. We experimentally demonstrate in \cref{sec:experiments}, that the \gls{rl} performance of an improved version of the \texttt{RAW}-\gls{vqc} policy is competitive. As the original definition is impractical for the upcoming discussions, we introduce a slight reformulation:
\begin{definition}[\texttt{RAW}-\gls{vqc}] \label{def:rawvqc}
	Given a \gls{vqc} acting on $n$ qubits, taking as input an \gls{rl} state $\bm{s} \in \mathbb{R}^n$, rotation angles $\bm{\theta} \in \left[ -\pi, \pi \right]^{\abs{\bm{\theta}}}$, scaling parameters $\bm{\lambda} \in \mathbb{R}^{\abs{\bm{\lambda}}}$, with $\Theta = (\bm{\theta},\bm{\lambda})$, such that it produces the quantum state $\ket{\psi_{\bm{s},\Theta}} = U_{\bm{s},\Theta}\ket{0}^{\otimes n}$, we define the \texttt{RAW}-\gls{vqc} policy:
	\begin{equation} \label{eq:rawvqc}
		\pi_{\Theta}(a \mid \bm{s}) = \expval{P_a}_{\bm{s},\Theta},
	\end{equation}
	where $\expval{P_a}_{\bm{s},\Theta} = \expval{ \psi_{\bm{s},\Theta} \left| P_a \right| \psi_{\bm{s},\Theta} }$ is the expectation value of a projector $P_a$. It must hold $P_a = \sum_{\ket{v} \in \mathcal{V}_a} \ket{v} \bra{v}$, with $\mathcal{V}_a \subseteq \mathcal{V}$, where $\mathcal{V} = \left\{ \ket{v_{0}}, \ket{v_{1}}, \cdots, \ket{v_{2^{n}-1}} \right\}$ is the set of eigenstates of an observable
    \begin{equation}
        \label{eq:spectral_decomposition}
	    O = \sum_{i=0}^{2^{n}-1} i \cdot \ket{v_i} \bra{v_i}.
	\end{equation}
    It must hold $\bigcup_{a \in \mathcal{A}} \mathcal{V}_a = \mathcal{V}$, and $\mathcal{V}_i \cap \mathcal{V}_j = \emptyset$ for all $i \neq j$.
\end{definition}

To be concise, the given reformulation is slightly more restrictive than the original one, due to the explicit designation of the eigenvalues in \cref{eq:spectral_decomposition}. This ensures, that by measuring eigenvalue $i$, we can directly conclude the post-measurement state to be $\ket{v_i}$. However, \cref{def:rawvqc} is completely equivalent in terms of all considerations and experiments carried out in \citeauthor{Jerbi_2021}.


\section{\label{sec:theory}Analysis of Action Decoding}

We now focus on the action decoding scheme of the \texttt{RAW}-\gls{vqc} policy. \cref{def:rawvqc} is instantiated such that the observable in \cref{eq:spectral_decomposition} can be efficiently replaced with only $1$-qubit Pauli operators. We start by measuring a $1$-qubit Pauli observable $\sigma_{n-1}$ on the uppermost qubit. The observed result is one of the eigenvalues $\mu_0 = +1$ or $\mu_1 = -1$, which is interpreted as the bit value $b_{n-1} = \frac{1-\mu}{2}$. The measured qubit is in the corresponding post-measurement state, while the other qubits have not been touched thus far. Now, we measure the Pauli observable $\sigma_{n-2}$ on the next to uppermost qubit, and proceed this way until we have measured all the qubits. As all the $1$-qubit Pauli observables on different qubits commute, the successive projections can be collected into one overall projection onto the respective basis state. The combined measurement result is the bitstring $b_{n-1} b_{n-2} \cdots b_{0}$, which is the binary expansion of $i$ in \cref{eq:spectral_decomposition}. We follow the convention that the most significant bit corresponds to the uppermost wire of \cref{fig:vqc}.

\subsection{\label{subsec:partitioning}Partitioning of Computational Basis States}

Since measurements can typically only be done in the energy eigenbasis of a qubit, we select the Pauli operators $\sigma_i$ to be all Pauli-$Z$ observables. This boils down to a measurement in the \emph{computational basis}, which for an $n$-qubit system is given by $\mathcal{V} = \lbrace \ket{0 \cdots 00}, \ket{0 \cdots 01}, \cdots, \ket{1 \cdots 11} \rbrace$. However, our techniques can also be applied to more general combinations of Pauli operators, as all share eigenvalues $+1$ and $-1$. Following \cref{def:rawvqc} this set has to be partitioned, i.e., $\mathcal{V}_a = \lbrace \ket{a_0}, \ket{a_1}, \cdots \rbrace \subseteq \mathcal{V}$ for action $a$. Let the prepared state be represented in the computational basis (using decimal notation) as $\ket{\psi_{\bm{s},\Theta}} = c_0 \ket{0} + c_1 \ket{1} + \cdots + c_{N-1} \ket{N-1}$, with $N = 2^n$. This allows the reformulation of \cref{eq:rawvqc} in terms of the absolute squared amplitudes of the prepared quantum state:
\begin{align}
	\pi_{\Theta} \left( a \mid \bm{s} \right) &= \expval{\psi_{\bm{s},\Theta} \left| \sum_{\ket{v} \in \mathcal{V}_a} \ket{v} \bra{v} \right| \psi_{\bm{s},\Theta}} \\
        &= \sum_{\ket{v} \in \mathcal{V}_a} \abs{c_v}^2.
        \label{eq:action_decoding_amplitudes}
\end{align}

Consequently, (as only Pauli observables are considered) it is sufficient to sum up the absolute squared amplitudes associated with the respective basis states to determine the policy. On quantum hardware it is possible to estimate the absolute squared value by executing the experiment multiple times. It corresponds to the probability of observing the eigenvalues associated with the respective basis states. With Pauli-$Z$ observables, measuring an eigenvalue of $i$ (which happens with probability $\abs{c_i}^2$) indicates the post-measurement state to be $\ket{i}$.  In practice the measurement result is the binary expansion of $i$, i.e., the bitstring $b_{n-1} b_{n-2} \cdots b_{0}$.

For interacting with the environment, the \gls{rl} agent selects an action according to the current policy $\pi_{\Theta}(a|s)$. Starting from \cref{eq:action_decoding_amplitudes}, it holds $\sum_{a \in \mathcal{A}} \sum_{\ket{v} \in \mathcal{V}_a} \abs{c_v}^2 = 1$. As all individual summands are non-negative, this defines a \glsentrylong{pdf}. Hence, it is sufficient to only measure the quantum state once. The agent decides for an action, based on which partition $\mathcal{V}_a$ the post-measurement state is contained within. For the parameter update, we must obtain $\nabla_{\Theta} \ln \pi_{\Theta}(a|\bm{s}) = \nabla_{\Theta} \expval{P_a}_{\bm{s},\Theta} / \expval{P_a}_{\bm{s},\Theta}$, for a trajectory of concrete instances of $\bm{s}$ and $a$. To estimate $\expval{P_a}_{\bm{s},\Theta}$ (and also $\nabla_{\Theta} \expval{P_a}_{\bm{s},\Theta}$), we need to determine for each post-measurement state $\ket{i}$, if it is an element of $\mathcal{V}_a$.

There is a caveat with the explicit representation of basis state partitionings~\cite{Jerbi_2021}. For larger systems, storing all $\mathcal{V}_a$ is infeasible, as the number of elements scales exponentially in $n$, independently of quantum or classical hardware (for $n=64$ qubits, there are $2^{64}$ possible bit strings, which would require $2^{64} \cdot \log_2(64)~ \text{bit} \approx 147.6 ~ \text{exabyte}$ of storage space). A solution would be to use a classically computable post-processing function based on the measurement outputs. One expects that the \gls{rl} performance strongly depends on the choice of post-processing function. In the sequel we introduce a measure of globality for post-processing functions and provide strong evidence that it correlates with \gls{rl} performance.

\subsection{\label{subsec:globality_measure}Action Decoding with Classical Post-Processing Function}

We denote the set of all bitstrings as $\bm{b} = b_{n-1}b_{n-2}\cdots b_1 b_0$ by $\mathcal{C}$, and partition it into disjoint, action-associated sets $\mathcal{C}_a$. We define a \emph{classical post-processing function} $f_{\mathcal{C}}: \lbrace 0,1 \rbrace^n \to \lbrace 0, 1, \cdots, \abs{\mathcal{A}}-1 \rbrace$, such that $f_{\mathcal{C}}(\bm{b}) = a$, iff $\bm{b} \in \mathcal{C}_a$ for an partitioning of $\mathcal{C}$. We can reformulate \cref{eq:action_decoding_amplitudes}:
\begin{align}
    \pi_{\Theta} (a \mid \bm{s}) =& \sum_{\bm{b} \in \left\{ 0,1 \right\}^n}^{f_{\mathcal{C}}(\bm{b})=a} \expval{ \psi_{\bm{s},\Theta} \mid \bm{b}} \expval{ \bm{b} \mid \psi_{\bm{s},\Theta}}\\
    \approx & \frac{1}{K} \cdot \sum_{k=0}^{K-1} \delta_{f_{\mathcal{C}} \left( \bm{b}^{(k)} \right) = a} \label{eq:classical_post_processing}
\end{align}
where $K \gg 1$ is the number of shots for estimating the expectation value (\cref{eq:classical_post_processing} becomes exact for $K \to \infty$), $\bm{b}^{(k)}$ is the bitstring observed in the $k$-th shot, and $\delta$ is an indicator function.


\subsubsection{\label{subsubsec:globality_measure}Extracted Information Defines Globality Measure}

In order to derive a quality measure for a specific post-processing function $f_\mathcal{C}$, we first define the notion of \emph{extracted information} for an observed bitstring $\bm{b}$:

\begin{definition}[extracted information]
    \label{def:extracted_info}
	Let $f_{\mathcal{C}}$ be a classical post-processing function, with a partitioning of the set of $n$-bit strings $\mathcal{C} = \bigcup_a \mathcal{C}_a$, for which $\mathcal{C}_i \cap \mathcal{C}_j = \emptyset$ for all $i \neq j$. Furthermore, $\bm{b} = b_{n-1}b_{n-2}\cdots b_1 b_0$ denotes an arbitrary bitstring. The extracted information $EI_{f_\mathcal{C}}(\bm{b}) \in \mathbb{N}$ is the minimum number of bits $b_i$ necessary, to compute $f_{\mathcal{C}}(\bm{b})$, i.e.\ assign $\bm{b}$ unambiguously to a set $\mathcal{C}_a$.
\end{definition}

An example of a valid partitioning associated with a $4$-qubit system and  $\abs{\mathcal{A}}=4$ is
\begin{eqnarray}
    \mathcal{C}_{a=0} &=& \left\{ 0000, ~0010, ~0100, ~0110 \right\}, \\
    \mathcal{C}_{a=1} &=& \left\{ 0001, ~0011, ~0101, ~0111 \right\}, \\
    \mathcal{C}_{a=2} &=& \left\{ 1000, ~1010, ~1101, ~1111 \right\}, \\
    \mathcal{C}_{a=3} &=& \left\{ 1001, ~1011, ~1100, ~1110 \right\}.
\end{eqnarray}
For assigning the bitstring $0111$ unambiguously to the correct set, i.e. $\mathcal{C}_{a=1}$, it is enough to consider only the first and the last bit. It is straightforward to see that it cannot work with less information. As all partitions contain the same number of elements, and we need to choose between $4$ actions, it requires at least $\log_2 (4) = 2$ bits of information. Therefore, the extracted information is given by $EI_{f_\mathcal{C}}(0111) = 2$.

To grant more expressivity to the defined measure, we average the extracted information over all possible bitstrings to get the \emph{globality measure}
\begin{equation}
    \label{eq:globality_measure}
	G_{f_\mathcal{C}} := \frac{1}{2^n} \sum_{\bm{b} \in \left\{ 0,1 \right\}^n} EI_{f_\mathcal{C}}(\bm{b}).
\end{equation}
The value of this measure describes the average amount of information (in bits) necessary to have an unambiguous distinction between the different actions. The concept is inspired by the reformulation of special policies using local and global observables in \cref{sec:local_global}.

A lower bound to $G_{f_{\mathcal{C}}}$ is intuitively given by $G_{f_{\mathcal{C}}} \geq \log_2 \left( \abs{\mathcal{A}} \right)$. The measure is trivially upper bounded by $G_{f_{\mathcal{C}}} \leq n$, which is in line with Holevo's theorem~\cite{Holevo_1973,Nielsen_2010}, in that no more than $n$ bits of classical information can be extracted from a $n$-qubit system in a single measurement. Evaluating \cref{eq:globality_measure} for the example above gives $G_{f_\mathcal{C}} = 2.5$, i.e., on average $2.5$ bits of information are necessary (see \cref{sec:example_ei} for the exact computation).

Evaluating \cref{eq:globality_measure} explicitly is infeasible for large $n$, as it requires averaging over $2^n$ elements. Furthermore, we are not aware of an efficient routine that determines the extracted information for an arbitrary bitstring. Nonetheless, one can define a post-processing function, which has maximal globality according to \cref{eq:globality_measure}. We discuss this construction in the next section.

\subsubsection{\label{subsubsec:optimal_partitioning}Constructing an Optimal Post-Processing Function}

As we demonstrate in \cref{sec:experiments}, the value of the introduced globality measure is strongly correlated with the \gls{rl} performance. It is not feasible to construct a post-processing function with optimal globality measure using a brute-force approach, as we argue in \cref{subsec:search_hard}. To circumvent this caveat, we construct an implicit partitioning $\mathcal{C}$, that gives rise to a post-processing function with provably optimal globality $G_{f_{\mathcal{C}}} = n$. The set of bitstrings $\bm{b} = b_{n-1} \cdots b_0$ associated with action $a$ is recursively defined as
\begin{equation}
    \label{eq:optimal_partitioning_recursive}
	\mathcal{C}_{[a]_2}^{(m)} = \left\{ \bm{b} \mid \bigoplus_{i=m}^{n-1} b_i = a_0 \wedge \bm{b} \in \mathcal{C}_{a_m \cdots a_2 \left( a_1 \oplus a_0 \right)}^{(m-1)} \right\}
\end{equation}
where $m = \log_2(M) - 1$ (with $M := \abs{\mathcal{A}}$) and $\left[ a \right]_2 = a_{m} \cdots a_{0}$ is the binary expansion of $a$. The base cases use a binary parity function on all bits:
\begin{align}
    \mathcal{C}_{[0]_2}^{(0)} =& \left\{ \bm{b} \mid \bigoplus_{i=0}^{n-1} b_i = 0 \wedge \bm{b} \in \left\{ 0, 1 \right\}^n \right\} \label{eq:optimal_partitioning_basis_0}\\
	\mathcal{C}_{[1]_2}^{(0)} =& \left\{ \bm{b} \mid \bigoplus_{i=0}^{n-1} b_i = 1 \wedge \bm{b} \in \left\{ 0, 1 \right\}^n \right\} \label{eq:optimal_partitioning_basis_1}
\end{align}

\noindent The construction in \cref{eq:optimal_partitioning_recursive} thus recursively splits \cref{eq:optimal_partitioning_basis_0,eq:optimal_partitioning_basis_1} by computing parity values of sub-strings, until the required number of groups is formed.

\begin{lemma}
    \label{lem:optimal_partitioning}
	Let an arbitrary \gls{vqc} act on an $n$-qubit state. The \texttt{RAW}-\gls{vqc} policy needs to distinguish between $M := \abs{\mathcal{A}}$ actions, where $M$ is a power of $2$, i.e., $m = \log_2(M) - 1 \in \mathbb{N}_{0}$. Using \cref{eq:optimal_partitioning_recursive,eq:optimal_partitioning_basis_0,eq:optimal_partitioning_basis_1} we define
    \begin{align}
		\pi_{\Theta}^{\text{glob}} \left( a \mid \bm{s} \right) &= \sum_{v \in \mathcal{C}_{[a]_2}^{(m)}} \expval{\psi_{\bm{s},\Theta} \mid v} \expval{v \mid \psi_{\bm{s},\Theta}} \\
    &\approx \frac{1}{K} \sum_{k=0}^{K-1} \delta_{f_{\mathcal{C}^{(m)}}(\bm{b}^{(k)}) = a}\label{eq:optimal_postprocessing}
	\end{align}
    where $K$ is the number of shots for estimating the expectation value, $\bm{b}^{(k)}$ is the bitstring observed in the $k$-th shot, and $\delta$ is an indicator function. The post-processing function is guaranteed to have the globality value $G_{f_{\mathcal{C}}} = n$.
\end{lemma}

The proof is deferred to \cref{sec:supplementary_optimal_partitioning}. This post-processing function defines the proposed \gls{qpg} algorithm in \cref{fig:overview}.

\input{figures/post-processing}

By construction of the recursive definition the action associated with a specific bitstring $f_{\mathcal{C}}(\bm{b})$ is given as
\begin{equation}
    \label{eq:optimal_post-processing_function}
    f_{\mathcal{C}}(\bm{b}) = \left[ b_0 \cdots b_{m-1} \left( \bigoplus_{i=m}^{n-1} b_i \right) \right]_{10}
\end{equation}
where $[\cdot]_{10}$ denotes the decimal representation of the respective bitstring. This representation does not require storing the partitioning $\mathcal{C}$ explicitly, which renders our approach feasible for large system sizes and action spaces.


\section{\label{sec:experiments}Experiments}

Our framework realizes different post-processing functions, associated with various globality measure values. The implementation is based upon the \texttt{qiskit} and \texttt{qiskit\_machine\_learning} libraries. If not stated differently, all experiments use the \texttt{StatevectorSimulator}, which assumes the absence of noise, and also eliminates sampling errors. The computations were executed on a CPU-cluster with $64$ nodes, each equipped with $4$ cores and $32$ GB of working memory. The experiments in \cref{subsec:exp_results,subsec:exp_fim} focus on the \texttt{CartPole} environment, while \cref{subsec:quantum_hardware,sec:experiments_alternated_setup} also consider \texttt{ContextualBandits} and \texttt{FrozenLake}, respectively. We establish conventions regarding the experimental setup and reproducibility in \cref{sec:experiments_alternated_setup}.

\subsection{\label{subsec:exp_results}RL Performance vs. Globality Measure}

The main experiment is conducted on the \texttt{CartPole-v0} environment~\cite{Brockman_2016} with a horizon of $200$ time-steps. The state space has a dimensionality of $4$, with all values scaled to be within $\left[ -1, 1 \right)$. The agent can take two actions, i.e., $\abs{\mathcal{A}}=2$.

All experiments in \cref{fig:partitioning} (apart from the \texttt{SOFTMAX}-\gls{vqc} policy $\pi^{sm}$, which uses a tensored Pauli-$Z$ measurement on all qubits~\cite{Jerbi_2021}) use the same architecture, only the post-processing function is modified. The two extreme cases are $G_{f_{\mathcal{C}}}=4.0$, constructed following \cref{lem:optimal_partitioning}, and $G_{f_{\mathcal{C}}}=1.0$, which extracts the lowest amount of information that is sufficient. We also experiment with $G_{f_{\mathcal{C}}}=2.0$ and $G_{f_{\mathcal{C}}}=3.0$, both of which can be expressed as a parity measurement on $2$ or $3$ qubits, respectively. A special case is $G_{f_{\mathcal{C}}}=3.5$, where the explicit partitioning $\mathcal{C}_{a=0}^{\text{3.5}} = \left\{ 1, ~3, ~5, ~6, ~9, ~10, ~12, ~15 \right\}$ and $\mathcal{C}_{a=1}^{\text{3.5}} = \left\{ 0, ~2, ~4 , ~7, ~8, ~11, ~13, ~14 \right\}$ is used.

Throughout all considerations, the \gls{rl} performance clearly benefits from a higher globality value of the underlying post-processing function. The convergence speed is improved, for example, the strategy learned by an agent with an underlying global post-processing function reaches optimal behavior after just $100$ episodes. This is clearly delayed for all other configurations. In fact, the policy with $G_{f_{\mathcal{C}}}=1.0$ is not able to learn optimal behavior, even after $5,000$ episodes. This can be partially addressed with deeper circuits (see \cref{subsec:deeper_circuits}). However, as circuit depth is very critical for \gls{nisq} devices, using optimal classical post-processing functions is crucial.

\cref{fig:partitioning} also depicts the performance of a \texttt{SOFTMAX}-\gls{vqc} policy~\cite{Jerbi_2021}. Interestingly, it performs better than \texttt{RAW}-\gls{vqc} with a globality value $\leq 3$, but is clearly inferior to the two fastest converging setups.

It is important to mention, that this overall behavior cannot only be observed in this concrete setup and environment. We obtained comparable results on \texttt{CartPole-v1}, which extends the horizon to $500$ steps. Results on further environments, namely \texttt{FrozenLake} and \texttt{ContextualBandits}, are provided in \cref{subsec:other_environments}.

\subsection{\label{subsec:exp_fim}Analysis of Effective Dimension and Fisher Information Spectrum}

\input{figures/effdim_spectrum}

For every machine learning task, two crucial factors are the \emph{expressibility} and \emph{trainability} of the used model. Tools for quantitative analysis, based on the \glsentrylong{fim} (FIM), have recently been proposed by \citeauthor{Abbas_2021}, and a comparative study of various quantum neural network architectures has been conducted in \citeauthor{Wilkinson_2022}. An adaption of those concepts to the \gls{rl} setup is deferred to \cref{sec:intro_fim}.

The expressibility of a model can be quantified using the \emph{effective dimension}~\cite{Abbas_2021}, which describes the variety of functions that can be approximated. A normalized version of the effective dimension for different \texttt{RAW}-\gls{vqc} policy setups is compared in \cref{subfig:effdim}. The expressive power of the respective model is proportional to the globality of the underlying post-processing function and also the \gls{rl} performance from \cref{fig:partitioning}. The policies with $G_{f_{\mathcal{C}}}=4$ and $G_{f_{\mathcal{C}}}=3.5$ pose an exception, as the respective effective dimensions coincide. We considered this statistical variance, as also the \gls{rl} performance varies only slightly. While a more expressive circuit provides no guarantee of better performance, a complex problem needs a model that is expressive enough, promoting the usage of post-processing functions with high globality.

Insights into the trainability of a model are possible by considering the spectrum of the \gls{fim}~\cite{Abbas_2021}, which captures the geometry of the parameter space. Trainability profits from a uniform spectrum, while distorted spectra are suboptimal. \cref{subfig:spectrum} depicts the Fisher information spectrum for post-processing functions with $G_{f_{\mathcal{C}}}=4.0$, $G_{f_{\mathcal{C}}}=3.0$, and $G_{f_{\mathcal{C}}}=1.0$. There is no clear difference between the Fisher information spectra associated with the ones with higher globality. The difference to the least-global configuration is more significant, where most of the eigenvalues are close to $0$. This implies that the parameter space is flat in most dimensions, making optimization difficult. Additionally, there are a few large eigenvalues, indicating a distorted optimization space. In absolute terms, the spectra for $\pi^{G_{f_{\mathcal{C}}}}_{=4}$ and $\pi^{G_{f_{\mathcal{C}}}}_{=3}$ are not uniform. However, in comparison to $\pi^{G_{f_{\mathcal{C}}}}_{=1}$, the eigenvalues are much more uniformly distributed, and also fewer outliers exist. This property becomes more significant when considering larger system sizes and circuit depths, see \cref{sec:additional_analysis_fisher_spectrum}. Hence, the globality value associated with a model correlates at least to some extent with the uniformity of the Fisher information spectrum, which is beneficial for trainability.

\subsection{\label{subsec:quantum_hardware}Training on Quantum Hardware}

To emphasize the practical relevance of our method, we conducted a experiment on actual quantum hardware. There is work on \gls{vqc}-based \gls{rl} that performs the training on classical hardware and then uploads the learned parameters to quantum hardware for testing~\cite{Chen_2020,Hsiao_2022}, where the trained models model can replicate the learned behavior on the hardware to some extent. We take a more involved approach and perform both training and testing on quantum hardware. To the best of our knowledge, this is the first investigation of \gls{vqc}-based \gls{rl} on quantum hardware that also includes the training routine.

\input{figures/hardware_new}

\input{tables/hardware}

We select an $8$-state \texttt{ContextualBandits} environment~\cite{Sutton_2018} for the experiment, which can be implemented with a $3$-qubit system. The employed hardware backend is the $5$-qubit device \texttt{ibmq\_manila v1.1.4}~\cite{IBMquantum_2023}. We slightly adapted the \gls{vqc} architecture and typical \gls{rl} feedback loop to make hardware usage feasible. First, we replaced the $CZ$ gates from \cref{fig:vqc} with $CX$ gates, due to the former one not being hardware-native (which would lead to decomposition and additional circuit complexity). Second, to reduce the number of hardware uploads, we used a batch size of $50$ trajectories. Consequently the gradients, and therefore also the parameter updates, are only computed for each $50$th time-step. Still, for a horizon of $1,500$ episodes, this adds up to overall $14,640$ expectation values that need to be estimated. With $1,024$ shots to estimate each one, close to $15$M circuits had to be evaluated per training run.


The training performance for different setups is displayed in \cref{fig:hardware}. We compare results on the hardware with and without matrix-free measurement error mitigation~\cite{Qiskit_2023,Nation_2021}. This is compared to results obtained from noise-free simulation on classical hardware. We also experimented with \texttt{qiskit} noise model instantiated with parameters sampled from the \texttt{ibmq\_manila} device -- the results were almost identical to the actual hardware. The noise-free simulation clearly produces the best results and is able to learn an basically optimal policy. While this is not the case for the experiments on hardware, there is still a clear improvement over the initial random policy. Hereby, as expected, the mitigated experiment (execution time approx. $360$ minutes over two \texttt{Qiskit Sessions}) improves upon the non-mitigated one (execution time approx. $150$ minutes in a single \texttt{Qiskit Session}). Interestingly, the performance of all three agents seems to saturate after about $1000$ episodes.

The testing results in \cref{tab:hardware_policy} clearly show that the (error-mitigated) hardware-trained agent is able to identify the optimal action for all $8$ states. The problem seems to be that the policy does not get ``peaky`` enough. We assume this is due to noise mainly induced by entangling gates $\varepsilon_{CX}$. While the original circuit uses only $15$ $CX$ gates, the transpiled versions average to about $27$, due to the sparse connectivity structure of the hardware device. It has to be noted that the overall length of the transpiled circuits stays approximately constant throughout all episodes. The re-calibration of the system after episode $750$ of the error-mitigated experiment (reducing $\varepsilon_{CX}$ from $0.70\%$ to $0.64\%$) did not cause a clear change in performance. We assume, that the convergence of the hardware agents towards a non-optimal expected reward is mainly caused by decoherence noise. To improve upon this, one can potentially use an architecture more adapted to the basis gate set and connectivity structure. Apart from that, more advanced error mitigation strategies~\cite{Giurgica_2020,Mari_2021} could be a suitable option.





\glsresetall

\section{Conclusion}

This paper analyzed the action decoding procedure of the \gls{qpg} algorithm originally proposed by \citeauthor{Jerbi_2021}. We proposed a hybrid routine combining measurements in the computational basis and a classical post-processing function. A newly developed globality measure for those functions showed a strong correlation with the \gls{rl} performance and model complexity measures. We provided a routine to implement a post-processing function that is optimal with respect to this measure -- which is also feasible for large action spaces. Compared to the original \texttt{RAW}-\gls{vqc}, as well as the \texttt{SOFTMAX}-\gls{vqc} policy~\cite{Jerbi_2021}, we achieve significant \gls{rl} performance improvements, with only negligible classical overhead. With this enhanced \gls{qpg} algorithm, we are able to execute the entire \gls{rl} training and testing routine on actual quantum hardware.

Our work focused on \gls{rl} routines, but in principle our findings can be extended to the realm of supervised and unsupervised learning. More concretely, the post-processing function for action selection can be reformulated to return the labels of a classification problem -- with reliable statements on transferability certainly requiring additional experiments. While we did not explicitly claim quantum advantage, the idea of constructing an environment based on the discrete logarithm from \citeauthor{Jerbi_2021} also holds for our approach.

\section*{Acknowledgements}

We wish to thank G.\ Wellein for his administrative and technical support of this work. The authors gratefully acknowledge the scientific support and HPC resources provided by the Erlangen National High Performance Computing Center (NHR@FAU) of the Friedrich-Alexander-Universit\"at Erlangen-N\"urnberg (FAU). The hardware is funded by the German Research Foundation (DFG).

We acknowledge the use of IBM Quantum services for this work. The views expressed are those of the authors, and do not reflect the official policy or position of IBM or the IBM Quantum team.

\textbf{Funding.} The research was supported by the Bavarian Ministry of Economic Affairs, Regional Development and Energy with funds from the Hightech Agenda Bayern via the project BayQS and by the Bavarian Ministry for Economic Affairs, Infrastructure, Transport and Technology through the Center for Analytics-Data-Applications (ADA Center) within the framework of “BAYERN DIGITAL II”.

M.\ Hartmann acknowledges support by the European Union’s Horizon 2020 research and innovation programme under grant agreement No 828826 ``Quromorphic'' and the Munich Quantum Valley, which is supported by the Bavarian state government with funds from the Hightech Agenda Bayern Plus.

\section*{Code Availability}

A repository with the framework to reproduce the main results of this paper is available at \url{https://gitlab.com/NicoMeyer/qpg_classicalpp}. Further information and data is available upon reasonable request.

%% file: figures/figureone.pdf_tex
\begingroup%
  \makeatletter%
  \providecommand\color[2][]{%
    \errmessage{(Inkscape) Color is used for the text in Inkscape, but the package 'color.sty' is not loaded}%
    \renewcommand\color[2][]{}%
  }%
  \providecommand\transparent[1]{%
    \errmessage{(Inkscape) Transparency is used (non-zero) for the text in Inkscape, but the package 'transparent.sty' is not loaded}%
    \renewcommand\transparent[1]{}%
  }%
  \providecommand\rotatebox[2]{#2}%
  \newcommand*\fsize{\dimexpr\f@size pt\relax}%
  \newcommand*\lineheight[1]{\fontsize{\fsize}{#1\fsize}\selectfont}%
  \ifx\svgwidth\undefined%
    \setlength{\unitlength}{693.14565103bp}%
    \ifx\svgscale\undefined%
      \relax%
    \else%
      \setlength{\unitlength}{\unitlength * \real{\svgscale}}%
    \fi%
  \else%
    \setlength{\unitlength}{\svgwidth}%
  \fi%
  \global\let\svgwidth\undefined%
  \global\let\svgscale\undefined%
  \makeatother%
  \begin{picture}(1,0.70082316)%
    \lineheight{1}%
    \setlength\tabcolsep{0pt}%
    \put(0,0){\includegraphics[width=\unitlength,page=1]{figureone.pdf}}%
    \put(0.41428023,0.24927335){\color[rgb]{0,0,0}\makebox(0,0)[t]{\lineheight{1.25}\smash{\begin{tabular}[t]{c}$U(\bm{\theta})$\end{tabular}}}}%
    \put(0.66105256,0.26416276){\color[rgb]{0,0,0}\rotatebox{-90}{\makebox(0,0)[t]{\lineheight{1.25}\smash{\begin{tabular}[t]{c}$0100\cdots1$\end{tabular}}}}}%
    \put(0.15355298,0.25084892){\color[rgb]{0,0,0}\makebox(0,0)[t]{\lineheight{1.25}\smash{\begin{tabular}[t]{c}$U\left(\bm{s}_t,\bm{\lambda}\right)$\end{tabular}}}}%
    \put(0.89263383,0.52654917){\color[rgb]{0,0,0}\makebox(0,0)[lt]{\lineheight{1.25}\smash{\begin{tabular}[t]{l}$a_t$\end{tabular}}}}%
    \put(0.16851149,0.48402081){\color[rgb]{0,0,0}\makebox(0,0)[lt]{\lineheight{1.25}\smash{\begin{tabular}[t]{l}$\bm{s}_t$\end{tabular}}}}%
    \put(0.36747043,0.64659846){\color[rgb]{0,0,0}\makebox(0,0)[lt]{\lineheight{1.25}\smash{\begin{tabular}[t]{l}environment\end{tabular}}}}%
    \put(0.07667918,0.02243939){\color[rgb]{0,0,0}\makebox(0,0)[lt]{\lineheight{1.25}\smash{\begin{tabular}[t]{l}variational quantum circuit\end{tabular}}}}%
    \put(0.12493517,0.48402081){\color[rgb]{0,0,0}\makebox(0,0)[rt]{\lineheight{1.25}\smash{\begin{tabular}[t]{r}state\end{tabular}}}}%
    \put(0.85401659,0.52876024){\color[rgb]{0,0,0}\makebox(0,0)[rt]{\lineheight{1.25}\smash{\begin{tabular}[t]{r}action\end{tabular}}}}%
    \put(0.87115814,0.16861479){\color[rgb]{0,0,0}\makebox(0,0)[t]{\lineheight{1.25}\smash{\begin{tabular}[t]{c}\\post-\\processing\end{tabular}}}}%
    \put(0.79504191,0.33653254){\color[rgb]{0,0,0}\makebox(0,0)[lt]{\lineheight{1.25}\smash{\begin{tabular}[t]{l}classical\end{tabular}}}}%
  \end{picture}%
\endgroup%

%% file: figures/post-processing.tex
\begin{figure*}
    \centering
    \begin{tikzpicture}
        \centering
        \begin{axis}[
        	xlabel=episode,
        	ylabel=average collected reward,
        	grid=both,
        	xmin=-30,xmax=1030,
            ymin=0,ymax=210,
            tick label style={font=\footnotesize},
            xtick={0,100,200,300,400,500,600,700,800,900,1000},
        	width=0.9\linewidth,
        	height=.25\textheight,
        	axis x line=bottom, axis y line=left, tick align = outside,
            legend columns=1,
        	legend style={/tikz/every even column/.append style={column sep=0.1cm, row sep=0.1cm},at={(1.13,0.6)},anchor=east,yshift=-5mm,font=\footnotesize},
        	no marks]
        	\addplot[line width=.5pt,solid,color=darkgreen!40] %
            	table[x=episode,y=glob,col sep=comma]{figures/data/data_partitioning.csv};
            \addplot[line width=.5pt,solid,color=blue!40] %
            	table[x=episode,y=loc35,col sep=comma]{figures/data/data_partitioning.csv};
            \addplot[line width=.5pt,solid,color=red!40] %
            	table[x=episode,y=loc3,col sep=comma]{figures/data/data_partitioning.csv};
            \addplot[line width=.5pt,solid,color=orange!40] %
            	table[x=episode,y=loc2,col sep=comma]{figures/data/data_partitioning.csv};
            \addplot[line width=.5pt,solid,color=violet!40] %
            	table[x=episode,y=loc1,col sep=comma]{figures/data/data_partitioning.csv};
            \addplot[line width=.5pt,solid,color=gray!40] %
            	table[x=episode,y=sm,col sep=comma]{figures/data/data_partitioning.csv};
            \addplot[line width=1pt,solid,color=darkgreen] %
            	table[x=episode,y=glob-avg,col sep=comma]{figures/data/data_partitioning.csv};
            \addplot[line width=1pt,solid,color=blue] %
            	table[x=episode,y=loc35-avg,col sep=comma]{figures/data/data_partitioning.csv};
            \addplot[line width=1pt,solid,color=red] %
            	table[x=episode,y=loc3-avg,col sep=comma]{figures/data/data_partitioning.csv};
            \addplot[line width=1pt,solid,color=orange] %
            	table[x=episode,y=loc2-avg,col sep=comma]{figures/data/data_partitioning.csv};
            \addplot[line width=1pt,solid,color=violet] %
            	table[x=episode,y=loc1-avg,col sep=comma]{figures/data/data_partitioning.csv};
            \addplot[line width=1pt,solid,color=gray] %
            	table[x=episode,y=sm-avg,col sep=comma]{figures/data/data_partitioning.csv};
            \addplot[color=black, domain=-25:1024, line width=1pt] {200.0};
            \addplot[color=white, domain=-25:1024, line width=5.5pt] {206.0};
            \addplot[dashed, color=black, domain=-25:1024, line width=1pt] {23.5};
            \legend{,,,,,,$\pi^{G_{f_{\mathcal{C}}}}_{=4}$,$\pi^{G_{f_{\mathcal{C}}}}_{=3.5}$,$\pi^{G_{f_{\mathcal{C}}}}_{=3}$,$\pi^{G_{f_{\mathcal{C}}}}_{=2}$,$\pi^{G_{f_{\mathcal{C}}}}_{=1}$,$\pi^{sm}$}
        \end{axis}
    \end{tikzpicture}
    \vspace{-3mm}
    \caption{\label{fig:partitioning}\gls{rl} training performance of \texttt{RAW}-\gls{vqc} policies with different post-processing functions on the \texttt{CartPole-v0} environment, averaged over $10$ independent runs and $20$ preceding timesteps (dark curves). A higher globality value correlates with a faster convergence to the optimal collected reward of $200$. For comparison, the performance of a \texttt{SOFTMAX}-\gls{vqc} is included as $\pi^{sm}$.}
\end{figure*}
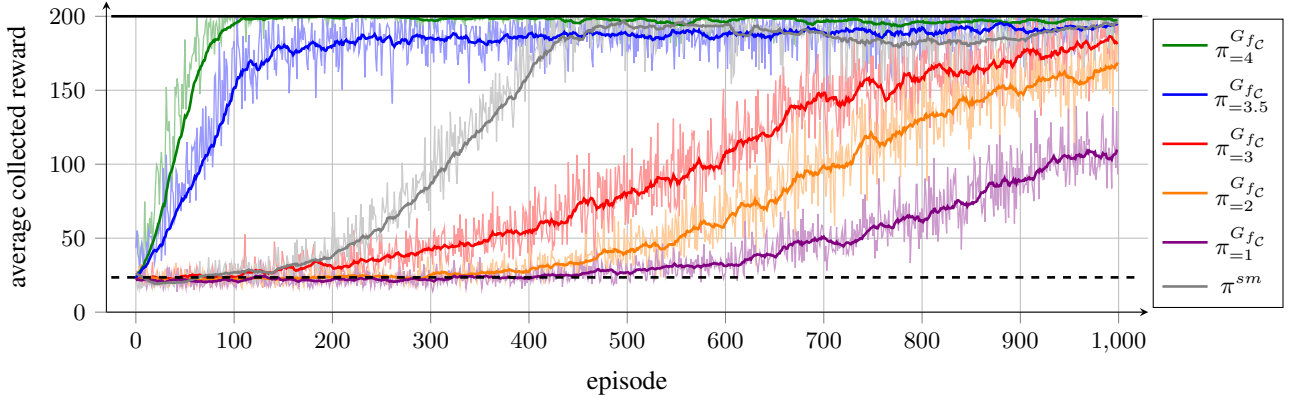

%% file: figures/effdim_spectrum.tex
\pgfplotsset{scaled x ticks=false}
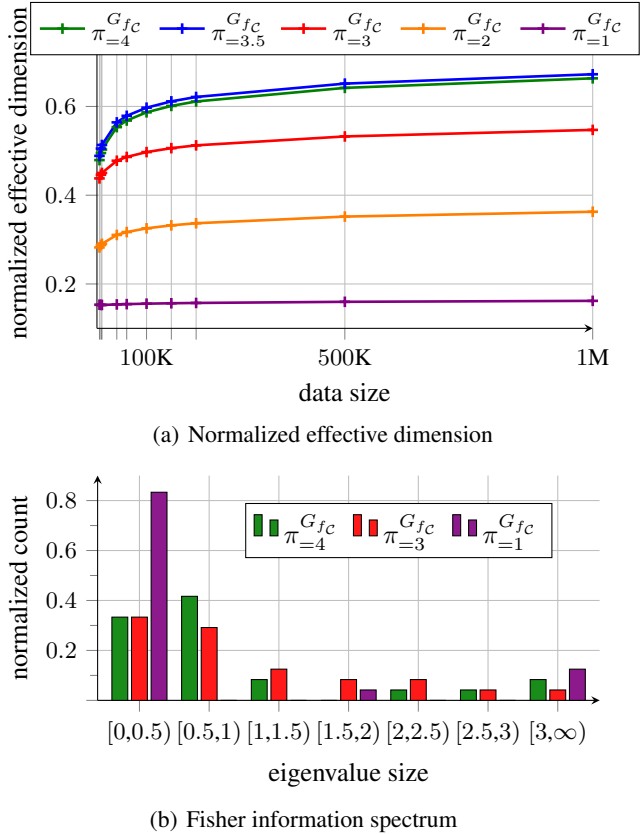
\begin{figure}[t]
    \centering
    \subfigure[Normalized effective dimension]{
        \centering
        \begin{tikzpicture}
        \begin{axis}[
                name=plot1,
        	xlabel={data size},
        	ylabel=normalized effective dimension,
        	grid=both,
        	xmin=-50,xmax=1000050,
                ymin=0.1,ymax=0.8,
                scaled ticks=false,
                tick label style={font=\footnotesize},
        	width=0.99\linewidth,
        	height=.25\textheight,
        	xtick={5000,8000,10000,40000,60000,100000,150000,200000,500000,1000000},
                xticklabels={,,,,,$100$K,,,$500$K,$1$M},
                ytick={0.2, 0.4, 0.6},
                yticklabels={$0.2$, $0.4$, $0.6$},
        	axis x line=bottom, axis y line=left, tick align = outside,
            legend columns=-1,
        	legend style={/tikz/every even column/.append style={column sep=0.1cm},at={(0.48,1.0)},anchor=south,yshift=-5mm,font=\footnotesize},
        	]
            \addplot[line width=1pt,solid,color=darkgreen,mark=+] %
            	table[x=data-size,y=g_4,col sep=comma]{figures/data/data_effdim.csv};
            \addlegendentry{$\pi^{G_{f_{\mathcal{C}}}}_{=4}$};
            \addplot[line width=1pt,solid,color=blue,mark=+] %
            	table[x=data-size,y=g_35,col sep=comma]{figures/data/data_effdim.csv};
            \addlegendentry{$\pi^{G_{f_{\mathcal{C}}}}_{=3.5}$};
            \addplot[line width=1pt,solid,color=red,mark=+] %
            	table[x=data-size,y=g_3,col sep=comma]{figures/data/data_effdim.csv};
            \addlegendentry{$\pi^{G_{f_{\mathcal{C}}}}_{=3}$};
            \addplot[line width=1pt,solid,color=orange,mark=+] %
            	table[x=data-size,y=g_2,col sep=comma]{figures/data/data_effdim.csv};
            \addlegendentry{$\pi^{G_{f_{\mathcal{C}}}}_{=2}$};
            \addplot[line width=1pt,solid,color=violet,mark=+] %
            	table[x=data-size,y=g_1,col sep=comma]{figures/data/data_effdim.csv};
            \addlegendentry{$\pi^{G_{f_{\mathcal{C}}}}_{=1}$};
        \end{axis}
        \end{tikzpicture}
        \label{subfig:effdim}
    }
    \\
    \subfigure[Fisher information spectrum]{
        \begin{tikzpicture}
            \begin{axis}[
            name=plot2,
            yshift={-0.5cm},
            width=\linewidth, height=.20\textheight,
            grid=major,
            tick label style={font=\footnotesize},
            minor y tick  num=1,
            xmin=0.2,xmax=3.8,
            ymin=0,ymax=0.9,
            xtick={0.5,1.0,1.5,2.0,2.5,3.0,3.5},
            xticklabels={$[0\text{,}0.5)$, $[0.5\text{,}1)$, $[1\text{,}1.5)$, $[1.5\text{,}2)$, $[2\text{,}2.5)$, $[2.5\text{,}3)$, $[3\text{,}\infty)$},
            ytick={0.0,0.2,0.4,0.6,0.8},
            yticklabels={,$0.2$, $0.4$, $0.6$, $0.8$},
            ylabel={normalized count},
            xlabel={eigenvalue size},
            ybar = .05cm,
            bar width = 6pt,
            axis x line=bottom, axis y line=left, tick align = outside,
            legend columns=-1,
            legend style={/tikz/every even column/.append style={column sep=0.1cm},at={(0.6,0.8)},anchor=south,yshift=-5mm}, 
            ]
            \addplot[fill=darkgreen!90,ybar,no marks,error bars/.cd, y dir=both, y explicit] coordinates {
                (0.5,8/24)
                (1.0,10/24)
                (1.5,2/24)
                (2.0,0)
                (2.5,1/24)
                (3.0,1/24)
                (3.5,2/24)
            };
            \addlegendentry{$\pi^{G_{f_{\mathcal{C}}}}_{=4}$};
            \addplot[fill=red!90,ybar,no marks,error bars/.cd, y dir=both, y explicit] coordinates {
                (0.5,8/24)
                (1.0,7/24)
                (1.5,3/24)
                (2.0,2/24)
                (2.5,2/24)
                (3.0,1/24)
                (3.5,1/24)
            };
            \addlegendentry{$\pi^{G_{f_{\mathcal{C}}}}_{=3}$};
            \addplot[fill=violet!90,ybar,no marks,error bars/.cd, y dir=both, y explicit] coordinates {
                (0.5,20/24)
                (1.0,0)
                (1.5,0)
                (2.0,1/24)
                (2.5,0)
                (3.0,0)
                (3.5,3/24)
            };
            \addlegendentry{$\pi^{G_{f_{\mathcal{C}}}}_{=1}$};
        \end{axis}
        \end{tikzpicture}
        \label{subfig:spectrum}
    }
\caption{\label{fig:effdim_spectrum} Quantities related to expressibility and trainability of a \gls{vqc}-based model. We estimate the \gls{fim} with $100$ random parameter sets for each of the $100$ random states $\bm{s}$. We draw the elements of each state from $\mathcal{N}\left( 0, 0.5 \right)$, which mimics the prior state distribution of the \texttt{CartPole} environment.}
\end{figure}

%% file: figures/hardware_new.tex
\begin{figure}
    \centering
    \begin{tikzpicture}
        \centering
        \begin{axis}[
        	xlabel=episode,
        	ylabel=expected reward,
        	grid=both,
        	xmin=-30,xmax=1530,
            ymin=-0.2,ymax=1.05,
            xtick={0, 500, 1000, 1500},
            tick label style={font=\footnotesize},
        	width=0.99\linewidth,
        	height=.25\textheight,
        	axis x line=bottom, axis y line=left, tick align = outside,
            legend columns=1,
        	legend style={/tikz/every even column/.append style={column sep=0.1cm},at={(0.64,0.17)},anchor=south,yshift=-5mm,font=\footnotesize},
        	no marks]
            \addplot[line width=.5pt,solid,color=darkgreen!40] %
            	table[x=episode,y=sim,col sep=comma]{figures/data/data_hardware_new.csv};
            \addplot[line width=.5pt,solid,color=orange!40] %
            	table[x=episode,y=hardware-mit,col sep=comma]{figures/data/data_hardware_new.csv};
            \addplot[line width=.5pt,solid,color=blue!40] %
            	table[x=episode,y=hardware,col sep=comma]{figures/data/data_hardware_new.csv};
            \addplot[line width=1pt,solid,color=darkgreen] %
            	table[x=episode,y=sim-avg,col sep=comma]{figures/data/data_hardware_new.csv};
            \addplot[line width=1pt,solid,color=orange] %
            	table[x=episode,y=hardware-mit-avg,col sep=comma]{figures/data/data_hardware_new.csv};
            \addplot[line width=1pt,solid,color=blue] %
            	table[x=episode,y=hardware-avg,col sep=comma]{figures/data/data_hardware_new.csv};
            \addplot[color=black, domain=-25:1524, line width=1pt] {1.0};
            \addplot[color=white, domain=-10:1524, line width=4.5pt] {1.03};
            \addplot[dashed, color=black, domain=-25:1524, line width=1pt] {0.0};
        \legend{,,,\texttt{StatevectorSimulator},\texttt{ibmq\_manila} (mitigated),\texttt{ibmq\_manila}}
        \end{axis}
    \end{tikzpicture}
    \vspace{-3mm}
    \caption{\label{fig:hardware}Training performance of a \texttt{RAW}-\gls{vqc} policy with maximum globality value on an $8$-state and $2$-action \texttt{ContextualBandits} environment. Two training runs are executed on quantum hardware (wo/ and w/ error mitigation).}
    
\end{figure}
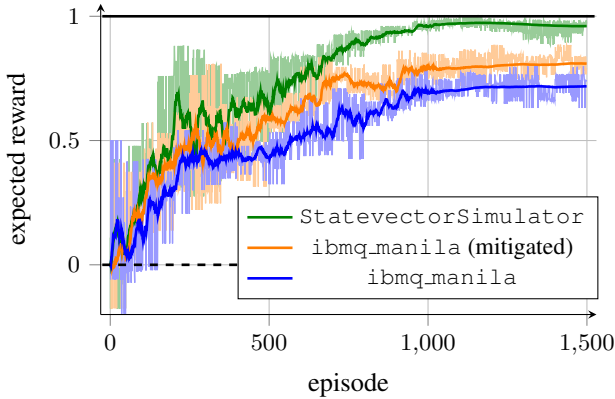

%% file: tables/hardware.tex
\begin{table}[t]
    \centering
    \caption{\label{tab:hardware_policy}Test results on the \texttt{ibmq\_manila} hardware device with error mitigation (orange curve in \cref{fig:hardware}).
    As for \texttt{ContextualBandits} the states are sampled uniformly at random, the optimal action is selected in approx. $92\%$ of all cases.\\}
        \begin{tabular}{c|cccc}
            \toprule
            & $s=0$ & $s=1$ & $s=2$ & $s=3$ \\
            \midrule
            optimal action & $a=0$ & $a=0$ & $a=1$ & $a=1$ \\
            $\pi(a=0 \mid s)$ & $\bm{0.93}$ & $\bm{0.91}$ & $0.09$ & $0.06$ \\
            $\pi(a=1 \mid s)$ & $0.07$ & $0.09$ & $\bm{0.91}$ & $\bm{0.94}$ \\
            \bottomrule
            \toprule
            & $s=4$ & $s=5$ & $s=6$ & $s=7$ \\
            \midrule
            optimal action & $a=0$ & $a=0$ & $a=1$ & $a=1$ \\
            $\pi(a=0 \mid s)$ & $\bm{0.94}$ & $\bm{0.91}$ & $0.09$ & $0.07$ \\
            $\pi(a=1 \mid s)$ & $0.06$ & $0.09$ & $\bm{0.91}$ & $\bm{0.93}$ \\
            \bottomrule
        \end{tabular}
\end{table}

%% file: appendix.tex
\section{\label{sec:restricted_softmax}Caveats of the \texttt{SOFTMAX}-VQC Policy}

In \cref{subsec:policies} we have stated, that using the \texttt{SOFTMAX}-\gls{vqc} policy is problematic w.r.t.\ circuit sampling complexity. This was not explicitly addressed in \citeauthor{Jerbi_2021}, where experimental results suggest that the \texttt{SOFTMAX}-\gls{vqc} policy formulation exhibits clearly superior performance in some simple benchmark environments, compared to the \texttt{RAW}-\gls{vqc} policy. This is partially explained by the argument, that this formulation has better abilities in dealing with the balancing of exploration and exploitation. More concretely, this trade-off can be influenced by the inverse temperature parameter $\beta$ in the \texttt{SOFTMAX}-\gls{vqc} policy equation
\begin{equation}
    \label{eq:softmax}
    \pi_{\bm{\lambda},\bm{\theta}}(a \mid \bm{s}) = \frac{e^{\beta \expval{O_a}_{\bm{s},\bm{\lambda},\bm{\theta}}}}{\sum_{a' \in \mathcal{A}} e^{\beta \expval{O_{a'}}_{\bm{s},\bm{\lambda},\bm{\theta}}}},
\end{equation}
\noindent where $\expval{O_a}_{\bm{s},\bm{\lambda},\bm{\theta}} := \expval{ \psi_{\bm{s},\bm{\lambda},\bm{\theta}} \left| O_a \right| \psi_{\bm{s},\bm{\lambda},\bm{\theta}}}$, and $O_a$ is some action-dependent observable.

There are two parts of the \gls{qpg} pipeline, in which this formulation has an undesirable circuit sampling complexity. This is especially troublesome for the currently existing \emph{nisq} devices, as every execution of a quantum circuit exhibits considerable costs. First of all, the action selection following \cref{eq:softmax} requires the estimation of $\abs{\mathcal{A}}$ different expectation values (unlike for the \texttt{RAW}-\gls{vqc}, where a single measurement in the computational basis is sufficient). Secondly, also the approximation of the log-policy gradients scales linearly with the number of actions $\abs{\mathcal{A}}$, as a different observable $O_a$ is selected for each action $a$ (unlike for the \texttt{RAW}-\gls{vqc}, where the estimation of only one expectation value is sufficient). In order to avoid this dependence of the scaling on the size of the action space, most experiments in Jerbi et al.~\cite{Jerbi_2021} use a fixed observable $O$ for all actions, which gets multiplied with some action-dependent classical weight $w_a$. We formalize this approach in the following:
\begin{definition}[\texttt{RESTRICTED}-\texttt{SOFTMAX}-\gls{vqc}]
    \label{def:restricted_softmax}
	Given a \gls{vqc} acting on $n$ qubits, taking as input a state $\bm{s} \in \mathbb{R}^n$, rotation angles $\bm{\theta} \in \left[ -\pi, \pi \right]^{\abs{\bm{\theta}}}$, and scaling parameters $\bm{\lambda} \in \mathbb{R}^{\abs{\bm{\lambda}}}$, such that it produces the quantum state $\ket{\psi_{\bm{s},\bm{\theta},\bm{\lambda}}} = U_{\bm{s},\bm{\theta},\bm{\lambda}}\ket{0}^{\otimes n}$, we define its associated \texttt{RESTRICTED}-\texttt{SOFTMAX}-\gls{vqc} policy as:
	\begin{equation}
	    \label{eq:restricted_softmax}
		\pi_{\Theta}(a \mid \bm{s}) = \frac{e^{\beta w_a \expval{O}_{\bm{s},\bm{\lambda},\bm{\theta}}}}{\sum_{a'} e^{\beta w_{a'} \expval{O}_{\bm{s},\bm{\lambda},\bm{\theta}}}}
	\end{equation}
	where $\expval{O}_{\bm{s},\bm{\lambda},\bm{\theta}} = \expval{\psi_{\bm{s},\bm{\lambda},\bm{\theta}} \mid O \mid \Psi_{\bm{s},\bm{\lambda},\bm{\theta}}}$ is the expectation value of the observable $O$, $w_a$ is a weight parameter associated with action $a$, and $\beta \in \mathbb{R}$ is an inverse-temperature parameter. $\Theta = \left( \bm{\theta}, \bm{\lambda}, \bm{w} \right)$ constitute all the trainable parameters of this policy.
\end{definition}

For completeness, we demonstrate in \cref{subsec:reduction_complexity}, that the circuit sampling complexity of this simplified version is only constant in the number of actions. A serious restriction of the \gls{rl} performance caused by this simplification us derived in \cref{subsec:structural_restriction}. In \cref{subsec:implications_softmax}, we extract some implications for the original \texttt{SOFTMAX}-\gls{vqc} policy from the findings.

\subsection{\label{subsec:reduction_complexity}Circuit Sampling Complexity of \texttt{RESTRICTED}-\texttt{SOFTMAX}-VQC Policy}

It follows directly from \cref{eq:restricted_softmax}, that the action selection procedure requires the estimation of only a single expectation value, namely $\expval{O}_{\bm{s},\bm{\lambda},\bm{\theta}}$. It is not directly obvious, if this reduction also translates to the gradient estimation, more concretely to the log-policy gradient required for \cref{eq:policy_gradient}. Following \cref{def:restricted_softmax}, the gradients w.r.t.\ $\bm{\theta}$ of a \texttt{RESTRICTED}-\texttt{SOFTMAX}-\gls{vqc} simplify to
\begin{eqnarray}
		&&\nabla_{\bm{\theta}} \ln \pi_{\Theta} \left(a \mid \bm{s} \right)
		\\&=& \beta \cdot \left( \nabla_{\bm{\theta}} w_a \expval{O}_{\bm{s},\bm{\theta},\bm{\lambda}}-\sum_{a'} \pi_{\Theta}(a' \mid \bm{s}) \cdot \nabla_{\bm{\theta}} w_{a'} \expval{O}_{\bm{s},\bm{\theta},\bm{\lambda}} \right) \\
		&=& \beta \cdot \nabla_{\bm{\theta}} \expval{O}_{\bm{s},\bm{\theta},\bm{\lambda}} \left( w_a - \sum_{a'} \pi_{\Theta}(a' \mid \bm{s}) \cdot w_{a'} \right).
\end{eqnarray}
A similar derivation holds for $\nabla_{\bm{\lambda}} \ln \pi_{\Theta} \left(a \mid \bm{s} \right)$. The gradient w.r.t.\ the weight associated with action $x$ is given as
\begin{eqnarray}
		&&\nabla_{w_{x}} \ln \pi_{\Theta} \left(a \mid \bm{s} \right) \\
		&=& \beta \cdot \left( \nabla_{w_{x}} w_a \expval{O}_{\bm{s},\bm{\theta},\bm{\lambda}} - \sum_{a'} \pi_{\Theta}(a' \mid \bm{s}) \cdot \nabla_{w_{x}} w_{a'} \expval{O}_{\bm{s},\bm{\theta},\bm{\lambda}} \right) \\
		&=& \beta \cdot \expval{O}_{\bm{s},\bm{\theta},\bm{\lambda}} \left( \delta_{a=x} - \pi_{\Theta}(x \mid \bm{s}) \right),
\end{eqnarray}
\noindent with $\delta_{a=x} = 1$, iff $a=x$, and $\delta_{a,x} = 0$ otherwise. Consequently, for both, action selection and gradient computation, only one observable has to be considered. This removes the dependence of the circuit sampling complexity on the number of actions, which must be avoided for environments with big action spaces.

\subsection{\label{subsec:structural_restriction}Structural Restriction of RL Performance for \texttt{RESTRICTED}-\texttt{SOFTMAX}-VQC}

While the previous considerations are quite promising when talking about circuit sampling complexity, there is also a serious drawback of the \texttt{RESTRICTED}-\texttt{SOFTMAX}-\gls{vqc} approach. More concretely, it is not suitable for problems with big action spaces, as most information is only contained in the classical weight parameters. This statement is concertized and proven in the following. First of all, we restrict our initial considerations to a specific type of \gls{rl} environment.
\begin{definition}[uniform environment] \label{def:uniform_env}
An environment $\mathcal{E}_{\mathcal{A}}$ is considered uniform, iff it is solved by a deterministic policy, which is expected to select each distinct action the same amount of times. More explicitly, let $\mathcal{S}_i$ denote a set of states from $\mathcal{S}$, with $\bigcup_{\mathcal{A}} \mathcal{S}_i = \mathcal{S}$ and $\mathcal{S}_i \cap \mathcal{S}_j = \emptyset$ for all $i \neq j$. Following the optimal policy $\pi_{*}$, each of the state sets must be equally likely to observe. With the notion of expected fraction of time spend in state $\bm{s}$ as $\mu(\bm{s})$ from \citeauthor{Sutton_2018}, this is stated as $\sum_{\bm{s} \in \mathcal{S}_i} \mu(\bm{s}) = \frac{1}{\abs{\mathcal{A}}}$. Let now $\mathcal{S}_i$ (with $i \in \lbrace 0, ~1, ~..., ~\abs{\mathcal{A}}-1 \rbrace$) be an arbitrary state set and $\bm{s}$ an arbitrary state from this set. It must hold, that $\pi_{*}\left( a_i \mid \bm{s} \right) = 1$, and consequently $\pi_{*}\left( a_j \mid \bm{s} \right) = 0$ for all $i \neq j$.
Hereby, the accuracy $ACC_{\pi}(\mathcal{E}_{\mathcal{A}})$ denotes the share of selected optimal actions in this environment, following policy $\pi$.
\end{definition}
A simple instance of such an environment can be constructed from a \texttt{ContextualBandits} scenario. Assume $8$ states and $4$ actions, where action $0$ is optimal for states from $\mathcal{S}_0 = \lbrace 0, 1 \rbrace$, action $1$ for states from $\mathcal{S}_2 = \lbrace 2, 3 \rbrace$, action $2$ for states from $\mathcal{S}_2 = \lbrace 4, 5 \rbrace$, and action $3$ for states from $\mathcal{S}_3 = \lbrace 6, 7 \rbrace$. As the states in the \texttt{ContextualBandits} environment are selected uniformly at random in every step, every state set is expectedly visited $\frac{1}{4}$ of the time. Additionally, the optimal action is different for all state sets, satisfying the conditions from \cref{def:uniform_env}.
\begin{lemma} \label{lem:restricted_softmax}
	Let $\pi$ be any \texttt{RESTRICTED}-\texttt{SOFTMAX}-\gls{vqc} policy, with w.l.o.g.\ $\expval{O}_{\bm{s}} \in \left[-1, 1 \right]$, for all $\bm{s} \in \mathcal{S}$. Given a uniform environment $\mathcal{E}_{\mathcal{A}}$, the performance of the model is upper bounded by
	\begin{equation} \label{eq:acc_restricted_softmax}
		ACC_{\pi} \left( \mathcal{E}_{\mathcal{A}} \right) \leq \frac{2}{\abs{\mathcal{A}}} \sum_{k=1}^{\abs{\mathcal{A}}/2} \frac{1}{k}.
	\end{equation}
\end{lemma}

\begin{proof}
	Assume for now, that $\expval{O}_{\bm{s}} \in \left(0, 1 \right]$ for all $\bm{s} \in \mathcal{S}$. Let $n = \abs{\mathcal{A}}$ denote the number of actions and the corresponding weights are w.l.o.g. ordered by
	\begin{equation} \label{eq:weight_ordering}
		w_0 \geq w_1 \geq \cdots \geq w_{n-2} \geq w_{n-1}.
	\end{equation}
	The statement $\pi \left( a_k \mid S_k \right) \leq \frac{1}{k}$ can be reformulated as $\sum_{\bm{s} \in \mathcal{S}_k} p_{\bm{s}} \pi \left( a_k \mid \bm{s} \right) \leq \sum_{\bm{s} \in \mathcal{S}_k} p_{\bm{s}} \frac{1}{k}$, where $p_{\bm{s}}$ denotes the probability of observing state $\bm{s}$ out of set $\mathcal{S}_k$. For the inequality to hold, it is sufficient that $\pi \left( a_k \mid \bm{s} \right) \leq \frac{1}{k}$ for all $\bm{s} \in \mathcal{S}_k$, which is proven by contradiction:
	\begin{align}
	    & \frac{e^{w_{k} \expval{O}_{\bm{s}}}}{\sum_{a'} e^{w_{a'} \expval{O}_{\bm{s}}}} &>&~~~~~~~~~~~~~~ \frac{1}{k} &&\\
 		\Leftrightarrow ~~~& (k-1) \cdot e^{w_{k} \expval{O}_{\bm{s}}} &>&~~~~~~~~~~~~~~ \sum_{a' \neq k} e^{w_{a'} \expval{O}_{\bm{s}}} && \\
 		\Leftrightarrow ~~~& k-1 &>&~~~~~~~~~~~~~~ \underbrace{\frac{e^{w_0 \expval{O}_{\bm{s}}}}{e^{w_{k} \expval{O}_{\bm{s}}}}}_{\geq 1} + \cdots + \underbrace{\frac{e^{w_{n-1} \expval{O}_{\bm{s}}}}{e^{w_{k} \expval{O}_{\bm{s}}}}}_{\geq 1} & \geq &~~~~~~~~~~ k-1 ~~~\lightning
	\end{align}
	The last step uses \cref{eq:weight_ordering} combined with the monotonicity of the exponential function. As $\mathcal{E}_{\mathcal{A}}$ is an uniform environment, for the described policy it holds $ACC_{\pi} \left( \mathcal{E}_{\mathcal{A}} \right) \leq \frac{1}{\abs{\mathcal{A}}} \sum_{k=1}^{\abs{\mathcal{A}}} \frac{1}{k}$. An improvement of this bound can be achieved by allowing $\expval{O}_{\bm{s}} \in \left[ -1, 1 \right]$. Multiplying with a negative value inverts the inequality chain from \cref{eq:weight_ordering} to $-w_0 \leq -w_1 \leq \cdots \leq -w_{n-2} \leq -w_{n-1}$, which introduces the missing factor of $0.5$ into \cref{eq:acc_restricted_softmax}. Hereby it is implicitly assumed that environment contains an even amount of actions, but it is straightforward to adapt the bound for the odd case. The case $\expval{O}_{\bm{s}} = 0$ does not lead to any improvement, as all actions will be selected with equal probability.
\end{proof}
This upper bound on performance makes the \texttt{RESTRICTED}-\texttt{SOFTMAX}-VQC unsuited for uniform environments with large action spaces. In fact, for $\abs{\mathcal{A}} \to \infty$ the accuracy converges to $0$. Already for $4$ actions, the accuracy is bounded by $\frac{2}{4} \left( \frac{1}{1} + \frac{1}{2} \right) = \frac{3}{4}$, experimental results on the \texttt{ContextualBandits} environment described above are depicted in \cref{fig:restricted_softmax}. We expect this result to hold in a weakened form for more general environments.

\input{figures/restricted_softmax}

\subsection{\label{subsec:implications_softmax}Implications for \texttt{SOFTMAX}-VQC Policy}

Directly following from \cref{lem:restricted_softmax}, one can make also a statement about the non-restricted \texttt{SOFTMAX}-\gls{vqc} policy:

\begin{corollary} \label{cor:softmax_vqc}
	Let $\pi$ be any \texttt{SOFTMAX}-\gls{vqc} policy and $\mathcal{E}_{\mathcal{A}}$ a uniform environment. In order to not impose any constraints following \cref{lem:restricted_softmax}, at most two actions can be associated with one unique observable. This results in a total requirement of $\lceil \abs{\mathcal{A}}/2 \rceil$ different observables.
\end{corollary}

Following \cref{cor:softmax_vqc}, the number of measured observables needs to increase linearly with the size of the action space, to avoid a bound on \gls{rl} performance. Of course, this is only a necessary and not sufficient condition for optimal performance of the model. We assume, that these results at least partially extend to more general environments.

\begin{corollary}
    \label{cor:softmax}
    In order to not put any performance constraints on the \texttt{SOFTMAX}-\gls{vqc} policy by construction, it has to incorporate $\mathcal{O}\left( \abs{\mathcal{A}} \right)$ \emph{different} observables $O_a$, i.e. the number has to scale linearly in the number of actions.
\end{corollary}

Unfortunately, a naive interpretation of \cref{cor:softmax} makes it impossible to circumvent the discussed bad scaling w.r.t.\ circuit sampling. Strategies to avoid the aforementioned scaling by a suitable choice of observables combined with classical post-processing might exist, but will not be considered in this paper. We are confident that the stated arguments and the performance advantage of the \texttt{RAW}-\gls{vqc} demonstrated in \cref{sec:experiments} provide adequate justification for the focus on this formulation.


\section{\label{sec:local_global}Action Decoding with Local and Global Observables}

\begin{figure*}[tb]
    \centering
    \subfigure[Measures only a $1$-qubit Pauli-$Z$ observable on the first qubit (associated with the \emph{local} policy in \cref{eq:policy_local}).]{
        \includegraphics[width=.28\linewidth]{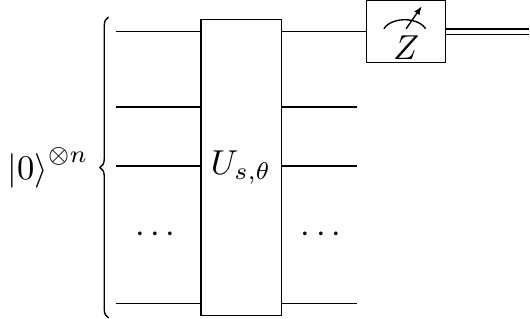}
    }
    \qquad
    \subfigure[Measures a tensored Pauli-$Z$ observable on all qubits (associated with the \emph{global} policy in \cref{eq:policy_global}). The depicted construction of $CX$ gates and a $1$-qubit Pauli-$Z$ measurement on an ancilla qubit defines a equivalent policy.]{
        \includegraphics[width=.6\linewidth]{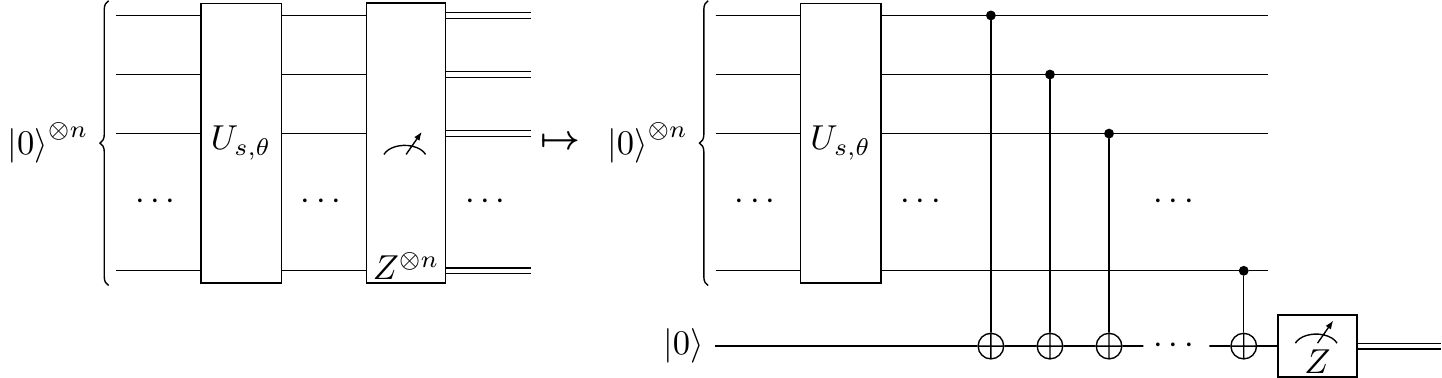}
    }
    \caption{Two types of observables, that give rise to different policy formulations.}
    \label{fig:local_global}
\end{figure*}

Continuing with the considerations from \cref{subsec:partitioning}, for an environment with $2$ actions, we have to define $\mathcal{V}_{a=0}$ and $\mathcal{V}_{a=1}$. Following \citeauthor{Jerbi_2021} for a $3$-qubit system yields:
\begin{eqnarray}
    	\mathcal{V}_{a=0}^{\text{loc}} &=& \left\{ \ket{000}, ~\ket{001}, ~\ket{010} , ~\ket{011} \right\} \\
    	\mathcal{V}_{a=1}^{\text{loc}} &=& \left\{ \ket{100}, ~\ket{101}, ~\ket{110} , ~\ket{111} \right\}
\end{eqnarray}
Looking at the above partitioning, all relevant information seems to be contained in the first qubit. For a measurement in the computational basis of the quantum state prepared by the \gls{vqc}, that returns the bitstring $b_2 b_1 b_0$, the action is given as $a = b_2$. In the more general case of $n$ qubits and $2$ actions, this consequently generalizes to $a^{\text{loc}} \gets b_{n-1}$, which can also be expressed as
\begin{equation}
    \label{eq:partitioning_local}
    \mathcal{V}_{a}^{\text{loc}} = \left\{ \ket{b_{n-1} b_{n-2} \cdots b_1 b_0} \mid a = b_{n-1} \right\}.
\end{equation}
At the other extreme, in the case of $2$ actions, we can select the action as $a^{\text{glob}} \gets \bigoplus_{i=0}^{n-1} b_i$, i.e., apply a parity function. For the sake of completeness, the corresponding eigenstate sets for $3$ qubits is
\begin{eqnarray}
    	\mathcal{V}_{a=0}^{\text{glob}} &=& \left\{ \ket{000}, ~\ket{011}, ~\ket{101} , ~\ket{110} \right\} \\
    	\mathcal{V}_{a=1}^{\text{glob}} &=& \left\{ \ket{001}, ~\ket{010}, ~\ket{100} , ~\ket{111} \right\},
\end{eqnarray}
and more generally for $n$ qubits:
\begin{equation}
    \label{eq:partitioning_global}
    \mathcal{V}_{a}^{\text{glob}} = \left\{ \ket{b_{n-1} b_{n-2} \cdots b_1 b_0} \mid a = \bigoplus_{i=0}^{n-1} b_i \right\}.
\end{equation}
So far, we interpreted the policy computation routine as the assignment of the post-measurement state to the containing partition. For the above $2$-action special cases in \cref{eq:partitioning_local,eq:partitioning_global}, it is possible to instead model the policy in terms of expectation values of observables:

\begin{eqnarray}
    \pi_{\Theta} \left( a|\bm{s} \right) &=& \sum_{v \in \mathcal{V}_a} \expval{\psi_{\bm{s},\Theta} \mid v} \expval{ v \mid \psi_{\bm{s},\Theta}} \\
	 &=& \frac{\sum_{v \in \mathcal{V}_a} \expval{\psi_{\bm{s},\Theta} \mid v} \expval{ v \mid \psi_{\bm{s},\Theta}}- \left( 1 - \sum_{v \in \mathcal{V}_a} \expval{\psi_{\bm{s},\Theta} \mid v} \expval{ v \mid \psi_{\bm{s},\Theta}} \right) + 1}{2} \\
	 &=& \frac{\sum_{v \in \mathcal{V}_a} \expval{\psi_{\bm{s},\Theta} \mid v } \expval{ v \mid \psi_{\bm{s},\Theta}} - \sum_{v \in \mathcal{V}_{\tilde{a}}} \expval{\psi_{\bm{s},\Theta} \mid v } \expval{ v \mid \psi_{\bm{s},\Theta}} + 1}{2},
\end{eqnarray}
where $\tilde{a}$ denotes the \emph{complement} action of $a$.

It is easy to check, that  measuring the observable $Z \otimes I^{\otimes n-1}$ returns the value $+1$, iff the post-measurement state lives in the space spanned by the elements of $\mathcal{V}_{a=0}^{\text{loc}}$, as defined in \cref{eq:partitioning_local}. Vice versa, in all other cases, the measurement outputs a value of $-1$. This simplifies the above equation to
\begin{equation}
    \label{eq:policy_local}
    \pi_{\Theta}^{\text{loc}} \left( a|\bm{s} \right) = \frac{ \left( -1 \right)^{a} \cdot \expval{\psi_{\bm{s},\Theta} \mid Z \otimes I^{\otimes n-1} \mid \psi_{\bm{s},\Theta}} + 1}{2}.
\end{equation}

Returning to the projector formalism from \cref{def:rawvqc}, this can alternatively also be expressed as $\pi_{\Theta}^{\text{loc}}(a | \bm{s}) = \expval{\psi_{\bm{s},\Theta} \left| \sum_{\bm{b} \in \lbrace 0,1 \rbrace^{n-1} } \ket{a}\ket{\bm{b}} \bra{a}\bra{\bm{b}} \right| \psi_{\bm{s},\Theta}}$, i.e., projections onto the two respective sub-spaces for $a=0$ and $a=1$.

A similar argument can be made for $\mathcal{V}_{a}^{\text{glob}}$, with the difference that the observable has to be $Z^{\otimes n}$, which gives
\begin{equation}
    \label{eq:policy_global}
    \pi_{\Theta}^{\text{glob}} \left( a \mid \bm{s} \right) = \frac{ \left( -1 \right)^{a} \cdot \expval{\psi_{\bm{s},\Theta} \mid Z^{\otimes n} \mid \psi_{\bm{s},\Theta}} + 1}{2}.
\end{equation}
As above, this could also be stated as projections onto the respective subspaces by reformulating $\pi_{\Theta}^{\text{glob}}(a \mid \bm{s}) = \expval{\psi_{\bm{s},\Theta} \left| \sum_{\bm{b} \in \lbrace 0,1 \rbrace^{n}}^{\oplus \bm{b} = a} \ket{\bm{b}} \bra{\bm{b}} \right| \psi_{\bm{s},\Theta}}$. Note, in general the post-measurement states for both formulations are different.

As also visualized in \cref{fig:local_global}, these approaches correspond to using a \emph{local} and \emph{global} observable, respectively. More precisely, for the left diagram, one should refer to a $1$-local measurement, as only a $1$-qubit observable is measured. For a $q$-local measurement, a $q$-qubit observable is measured on some subset containing $q$ out of the $n$ qubits. When $q=n$, we arrive at a global observable, as shown in the right part of the diagram. It seems plausible, that as $q$ approaches $n$, the measurement can be thought of as becoming more and more global.
The two setups displayed in \cref{fig:local_global} are only the edge cases. Let us assume a post-processing function that decides on an action based on the parity of the first $q$ bits of the reconstructed bitstring. This can be expressed as
\begin{equation}
    \label{eq:policy_q_local}
    \pi_{\Theta}^{\text{q-loc}} \left( a \mid \bm{s} \right) = \frac{ \left( -1 \right)^{a} \cdot \expval{\psi_{\bm{s},\Theta} \mid Z^{\otimes q} \otimes I^{\otimes n-q} \mid \psi_{\bm{s},\Theta}} + 1}{2},
\end{equation}
where $\pi_{\Theta}^{\text{n-loc}}$ is equivalent to $\pi_{\Theta}^{\text{glob}}$. Note that in all case only one bit of information is necessary to select one of the two actions as $\log_2(2) = 1$. Still, experiments in \cref{sec:experiments} suggest, that the \gls{rl} performance benefits from more global measurements.

This type of formulation removes the need to explicitly store partitionings and hence avoids the caveats described above. Unfortunately, this analysis only works for some special cases, i.e., when all the information is compressed into a subset of the qubits. Also, it does not generalize to larger action spaces, as the derivation of \cref{eq:policy_local,eq:policy_global,eq:policy_q_local} had to assume $\abs{\mathcal{A}} = 2$. Apart from that, the distinction between local and global observables in the considerations above is slightly incorrect. Instead of performing a global measurement to evaluate \cref{eq:policy_global}, we could get the same result by measuring a $1$-qubit Pauli-$Z$ observable on an ancilla qubit, as visualized in the right part of \cref{fig:local_global}. This ancilla qubit is initialized to $\ket{0}$, and after the evolution of the system with $U_{\bm{s},\Theta}$, it interacts via a $CX$-gate with each of the $n$ original qubits.


\section{\label{sec:example_ei}Extended Example on Extracted Information and Globality Measure}

This appendix deals with a closer analysis of the extracted information and globality measure of the partitioning example introduced in \cref{subsubsec:globality_measure}:
\begin{align}
    \mathcal{C}_{a=0} &= \left\{ 0000, ~0010, ~0100, ~0110 \right\} \label{eq:C_example_1} \\
    \mathcal{C}_{a=1} &= \left\{ 0001, ~0011, ~0101, ~0111 \right\} \label{eq:C_example_2} \\
    \mathcal{C}_{a=2} &= \left\{ 1000, ~1010, ~1101, ~1111 \right\} \label{eq:C_example_3} \\
    \mathcal{C}_{a=3} &= \left\{ 1001, ~1011, ~1100, ~1110 \right\} \label{eq:C_example_4}
\end{align}

As the system is really small, it is straightforward to determine the extracted information following \cref{def:extracted_info} for all $2^4 = 16$ bitstrings, by just considering all bit combinations. However, as already discussed previously, this is not feasible in the general case.

\input{tables/extracted_info}

With this work done, it is straightforward to compute the associated globality measure following \cref{eq:globality_measure} as
\begin{eqnarray}
    G_{f_{\mathcal{C}}} &=& \frac{1}{2^4} \sum_{\bm{b} \in \lbrace 0, 1 \rbrace^4} EI_{f_{\mathcal{C}}}(\bm{b}) \\
    &=& \frac{8 \cdot 2 + 8 \cdot 3}{16} = 2.5.
\end{eqnarray}
\noindent The result reads itself as that on average $2.5$ bit of information is necessary to make an unambiguous action assignment. This is obviously above the minimum value of $G^{\min}_{f_{n=4}} = \log_2 (4) = 2$, but well below the optimum of $G^{\max}_{f_{n=4}} = n = 4$. In \cref{subsec:application_example} we demonstrated how a post-processing function with optimal globality measure can be constructed for this setup.


\section{\label{sec:supplementary_optimal_partitioning}Supplementary Material on Construction of an Optimal Partitioning}

As discussed throughout \cref{subsec:globality_measure}, it is not trivial to come up with a bitstring partitioning, whose associated post-processing function is optimal w.r.t.\ the globality measure in \cref{eq:globality_measure}. However, this property is highly desirable, as it strongly correlates with \gls{rl} performance, as demonstrated in \cref{subsec:exp_results}.

\subsection{\label{subsec:search_hard}Direct Search for an Optimal Post-Processing Function is Infeasible}

Unfortunately, the number of possible partitionings is too large to perform any form of unstructured search. In fact, the number increases super-exponentially with the number of qubits $n$. To give some proportion, for $M$ actions and $N=2^n$ bitstrings, there are $N!/ \left[ M! \left(\frac{N}{M}! \right)^M \right]$ possibilities, where it is assumed that $M$ is a power of $2$, and $\mathcal{C}$ is split into sets of equal size. As some point of reference, this evaluates to approximately $2.8 \cdot 10^{34}$ potential partitionings for $N=2^6=64$ (i.e., a \gls{vqc} with $6$ qubits) and $M=4$, which corresponds to a small quantum system, even for \gls{nisq} standards. Lastly, a post-processing function with an underlying random partitioning is very unlikely to have a high globality value close to $G_{f_{\mathcal{C}}}=n$, as shown in \cref{fig:globality_dist}.

\input{figures/partitioning_histogram}

\subsection{\label{subsec:proof_optimality}Proof of Optimality for Proposed Construction}

In the main section, we proposed an approach to recursively construct a partitioning, for which the post-processing function is provably optimal w.r.t.\ the globality measure. For convenience, we restate \cref{lem:optimal_partitioning} below:

\begin{no-lemma}
	Let an arbitrary \gls{vqc} act on an $n$-qubit system. The \texttt{RAW}-\gls{vqc} policy needs to distinguish between $M := \abs{\mathcal{A}}$ actions, where $M$ is a power of $2$, i.e., $m = \log_2(M) - 1 \in \mathbb{N}_{0}$. Using the recursive definition from \cref{eq:optimal_partitioning_recursive,eq:optimal_partitioning_basis_0,eq:optimal_partitioning_basis_1}, one can define
    \begin{align*}
		\pi_{\Theta}^{\text{glob}} \left( a \mid \bm{s} \right) &= \sum_{v \in \mathcal{C}_{[a]_2}^{(m)}} \expval{\psi_{\bm{s},\Theta} \mid v} \expval{v \mid \psi_{\bm{s},\Theta}} \\
    &\approx \frac{1}{K} \sum_{k=0}^{K-1} \delta_{f_{\mathcal{C}^{(m)}}(\bm{b}^{(k)}) = a},
	\end{align*}
    where $K$ is the number of shots, $\bm{b}^{(k)}$ is the bitstring observed in the $k$-th experiment, and $\delta$ is an indicator function. The post-processing function associated with this policy is guaranteed to have the highest possible globality measure value $G_{f_{\mathcal{C}}} = n$.
\end{no-lemma}
\begin{proof}
	The proof uses induction over $m$. The base case for $m=0$ for \cref{eq:optimal_partitioning_basis_0,eq:optimal_partitioning_basis_1} is trivial, as it corresponds to the previous considerations from \cref{eq:partitioning_global,eq:policy_global}. The induction step $m \to m+1$ needs to consider the two sets, into which $\mathcal{C}_{[a]_2}^{(m)}$ gets decomposed by inversely applying \cref{eq:optimal_partitioning_recursive}:
	    \begin{eqnarray}
		    \mathcal{C}_{a_m \cdots a_1 0 a_0}^{(m+1)} &=& \left\{ \bm{b} = b_{n-1} b_{n-2} \cdots b_1 b_0 \mid \highlight{\bigoplus_{i=m+1}^{n-1} b_i = a_0} \wedge \bm{b} \in \mathcal{C}_{\underbrace{a_m \cdots a_1 \left( 0 \oplus a_0 \right)}_{[a]_2}}^{(m)} \right\} \\
		    \mathcal{C}_{a_m \cdots a_1 1 \tilde{a}_0}^{(m+1)} &=& \left\{ \bm{b} = b_{n-1} b_{n-2} \cdots b_1 b_0 \mid \highlight{\bigoplus_{i=m+1}^{n-1} b_i = \tilde{a}_0} \wedge \bm{b} \in \mathcal{C}_{\underbrace{a_m \cdots a_1 \left( 1 \oplus \tilde{a}_0 \right)}_{[a]_2}}^{(m)} \right\},
	    \end{eqnarray}
	\noindent where $\tilde{a}_0$ indicates a bitflip. The property of maximum globality of those two sets w.r.t.\ to any partition $\mathcal{C}^{(m+1)}$ is directly transferred by the induction requirement, as the information of all $m$ least-significant bits is required for that distinction. The marked parts in the above equations highlight, that also the remaining $n-m$ most-significant bits are required for deciding between action $\left[ a_m \cdots a_1 0 a_0 \right]_{10}$ and $\left[ a_m \cdots a_1 1 \tilde{a}_0 \right]_{10}$. Consequently, it is necessary to always consider the entire bitstring, which implies a globality measure value of $G_{f_{\mathcal{C}}} = n$.
\end{proof}

\subsection{\label{subsec:application_example}Example of Optimal Partitioning}

To get some intuition for the construction arising from \cref{lem:optimal_partitioning}, we construct a global partitioning for a setup with $n=4$ qubits and $M=4$ actions. Consequently, the partitions for the respective actions can be derived recursively with $m = \log_2(4)-1 = 1$.

\begin{eqnarray}
    \mathcal{C}^{(1)}_{[0]_2 = 00} &=& \left\{ \bm{b} \mid \highlight{b_3 \oplus b_2 \oplus b_1 = 0} \wedge \bm{b} \in \mathcal{C}^{(0)}_{0 \oplus 0 = [0]_2} \right\} = \left\{ \fcolorbox{black}{white}{000}\fcolorbox{white}{white}{0}, \fcolorbox{black}{white}{011}\fcolorbox{white}{white}{0}, \fcolorbox{black}{white}{101}\fcolorbox{white}{white}{0}, \fcolorbox{black}{white}{110}\fcolorbox{white}{white}{0} \right\} \\
    \mathcal{C}^{(1)}_{[3]_2 = 11} &=& \left\{ \bm{b} \mid \highlight{b_3 \oplus b_2 \oplus b_1 = 1} \wedge \bm{b} \in \mathcal{C}^{(0)}_{1 \oplus 1 = [0]_2} \right\} = \left\{ \fcolorbox{black}{white}{001}\fcolorbox{white}{white}{1}, \fcolorbox{black}{white}{010}\fcolorbox{white}{white}{1}, \fcolorbox{black}{white}{100}\fcolorbox{white}{white}{1}, \fcolorbox{black}{white}{111}\fcolorbox{white}{white}{1} \right\}
\end{eqnarray}

In both cases, after one step of recursion, the base case is reached.
\begin{equation}
    \mathcal{C}^{(0)}_{[0]_2} = \lbrace \bm{b} \mid b_3 \oplus b_2 \oplus b_1 \oplus b_0 = 0 \rbrace
\end{equation}

The construction for the remaining two partitions works totally equivalent:
\begin{eqnarray}
    \mathcal{C}^{(1)}_{[1]_2 = 01} &=& \left\{ \bm{b} \mid \highlight{b_3 \oplus b_2 \oplus b_1 = 1} \wedge \bm{b} \in \mathcal{C}^{(0)}_{0 \oplus 1 = [1]_2} \right\} = \left\{ \fcolorbox{black}{white}{001}\fcolorbox{white}{white}{0}, \fcolorbox{black}{white}{010}\fcolorbox{white}{white}{0}, \fcolorbox{black}{white}{100}\fcolorbox{white}{white}{0}, \fcolorbox{black}{white}{111}\fcolorbox{white}{white}{0} \right\} \\
    \mathcal{C}^{(1)}_{[2]_2 = 10} &=& \left\{ \bm{b} \mid \highlight{b_3 \oplus b_2 \oplus b_1 = 0} \wedge \bm{b} \in \mathcal{C}^{(0)}_{1 \oplus 0 = [1]_2} \right\} = \left\{ \fcolorbox{black}{white}{000}\fcolorbox{white}{white}{1}, \fcolorbox{black}{white}{011}\fcolorbox{white}{white}{1}, \fcolorbox{black}{white}{101}\fcolorbox{white}{white}{1}, \fcolorbox{black}{white}{110}\fcolorbox{white}{white}{1} \right\} 
\end{eqnarray}

Also here the recursion only has depth one and makes use of the other base case.
\begin{equation}
    \mathcal{C}^{(0)}_{[1]_2} = \lbrace \bm{b} \mid b_3 \oplus b_2 \oplus b_1 \oplus b_0 = 1 \rbrace
\end{equation}

We are guaranteed by \cref{lem:optimal_partitioning} that $G_{f_{\mathcal{C}}} = 4$, which is also easy to check for this example. It is straightforward to continue from here, i.e.\ repeatedly split the partitions to account for $8$ or $16$ actions. For example, the partition for action $a=5$ for $\abs{\mathcal{A}}=8$ actions is constructed as $\mathcal{C}^{(2)}_{[5]_2 = 101} = \left\{ \bm{b} \mid b_3 \oplus b_2 = 1 \wedge \bm{b} \in \mathcal{C}^{(1)}_{1 \left( 0 \oplus 1 \right) = [3]_2} \right\} = \left\{ 0101, 1001 \right\}$. \cref{eq:optimal_post-processing_function} allows to determine the class-association of an arbitrary bitstring without explicitly constructing the partitioned sets. For an setup with $8$ actions, i.e.\ $m = \log_2{8}-1 = 2$ we get
\begin{equation}
    f_{\mathcal{C}}(\bm{b}) = \left[ 1 0 \left( 1 \oplus 0  \right) \right]_{10} = 5,
\end{equation}
\noindent which correctly identifies $1001$ to be an element of partition $\mathcal{C}^{(2)}_{[5]_2}$.


\section{\label{sec:intro_fim}Analysis of Fisher Information for Reinforcement Learning Setup}

Different metrics for analyzing the \emph{expressibility} and \emph{trainability} for an \gls{qml} model have recently proposed by \citeauthor{Abbas_2021}. This work interprets the \gls{vqc} as a statistical model with the joint distribution $p_{\Theta}(x, y) = p_{\Theta}(y \mid x) p(x)$ for data pairs $(x,y)$. The prior $p(x)$ describes the distribution of input states, while $p_{\Theta}(y \mid x)$ gives the relationship between input and output of the model. We need to adapt this notion for the \gls{rl} setup, which results in $p_{\pi_{\Theta}}(\bm{s},a) = \pi_{\Theta} (a \mid \bm{s}) p_{\pi_{\Theta}}(\bm{s})$, where the prior state distribution $p_{\pi_{\Theta}} : \mathcal{S} \to \left[ 0,1 \right]$ depends on the policy in most environments. However, for practical reasons, we drop this explicit dependence, while still trying to imitate the distribution of states one would get by following the policy. Due to the loss of statistical independence of data samples (which is one of the most distinguishing features between supervised \gls{ml} and \gls{rl}), it is presently unclear to which extent the effective dimension can still be used as a well-defined capacity measure for the function approximation architecture we employ in our work. However, classification and action selection present related tasks. Therefore we assume, that the effective dimension of the \gls{vqc} circuit architecture (interpreted as a classifier for a supervised \gls{ml} task) at least serves as a rough indicator of its capacity for policy approximation in the \gls{rl} context. For the \texttt{ContextualBandits} environment, where the consecutive states are sampled independently at random (i.e.\ independent of the policy), the notion is exact.

The key component of the proposed measures~\cite{Abbas_2021} is the \gls{fim} $F(\Theta) \in \mathbb{R}^{\abs{\Theta} \times \abs{\Theta}}$. Briefly going into theoretical details, it is a Riemannian metric, given rise to by the Riemannian space formed by $\Phi$. From the full parameter space $\Phi$ each individual parameter set $\Theta$ is draw. For \glspl{vqc}  consisting mainly of parameterized rotations $\Phi \subset \left[-\pi, \pi \right)^{\abs{\Theta}}$ is a reasonable choice. In practice, it is necessary to approximate the \gls{fim} by the empirical \gls{fim}. With samples drawn independently and identically distributed from the ground truth $(s_i,a_i)^{k}_{i=1} \sim p_{\pi_{\Theta}}(s, a)$, it is given as
\begin{equation}
    \label{eq:empirical_fim}
	\tilde{F}_{k}(\Theta) = \frac{1}{k} \sum_{i=1}^{k} \left[ \nabla_{\Theta} \ln p_{\pi_{\Theta}}(s_i, a_i) \nabla_{\Theta} \ln p_{\pi_{\Theta}}(s_i, a_i)^t \right].
\end{equation}

An alternate formulation from \citeauthor{Sequeira_2022} drops the dependence on the prior state distribution. More concretely, this reduces \cref{eq:empirical_fim} to $\tilde{F}_{k}(\Theta)=\frac{1}{k} \sum_{i=1}^{k} \left[ \nabla_{\Theta} \ln \pi_{\Theta} \left( a_i \mid s_i \right) \nabla_{\Theta} \ln \pi_{\Theta} \left( a_i \mid s_i \right)^t \right]$. However, we assume that keeping the potentially inaccurate prior state information still should be beneficial.

\subsection{Expressive Power of the VQC-Model}

The \gls{fim} $F(\Theta)$, and for sufficiently high $k$ also $\tilde{F}_{k}(\Theta)$, captures the geometry of the parameter space, which allows us to define a measure for the expressibility of a given model. More explicitly, the \emph{effective dimension}~\cite{Abbas_2021} quantifies the variety of functions that can be approximated with a given model. For the parameter space $\Phi \subset \mathbb{R}^{\abs{\Theta}}$, the effective dimension of the statistical model $\mathcal{M}_{\Theta}$ associated with the \gls{vqc} setup can be defined as
\begin{equation}
    \label{eq:effdim}
    ed_{n}(\mathcal{M}_{\Phi}) := 2 \frac{\ln \left( \frac{1}{V_{\Phi}} \int_{\Phi} \sqrt{\det \left( I_{\abs{\Theta}} + \frac{n}{2 \pi \ln n} \hat{F}(\Theta)  \right)} d \Theta \right)}{\ln \left( \frac{n}{2 \pi \ln n} \right)}.
\end{equation}

The \gls{fim} $\hat{F}(\Theta)$ is normalized, such that $\frac{1}{V_{\Phi}} \int_{\Phi} \mathrm{tr} \left( \hat{F}(\Theta) \right) d \Theta = \abs{\Theta}$ holds, where $V_{\Phi} := \int_{\Phi} d \Theta$ denotes the volume of the parameter space. In practice, the (normalized) \gls{fim} is replaced with the respective empirical formulation. The parameter $n$ determines the effective resolution of the parameter space. Although the effective dimension is not guaranteed to increase monotonically with this data size, it is usually the case for \gls{ml} tasks.

\subsection{Trainability of the VQC-Model}

The spectrum of the \gls{fim}, i.e.\ its eigenvalue distribution, provides insights into the trainability of a model~\cite{Abbas_2021}. In general, trainability profits from an uniform spectrum, while a distorted one is suboptimal. As noted previously, in practice the \gls{fim} is replaced with its empirical approximation in \cref{eq:empirical_fim}.


\section{\label{sec:experiments_alternated_setup}Supplementary Experiments and Conventions}

We now establish some conventions regarding experimental setup and reproducibility. Initially, we experimented with a variety of different hyperparameter settings. Overall, the qualitative observations were quite stable throughout. For the results reported in this paper we fixed most hyperparameters, in order to make results more comparable. However, sometimes slight deviations are necessary to improve performance, which is typical for \gls{rl} and also \gls{qrl}~\cite{Franz_2022}. To start with, all experiments on the \texttt{CartPole-v0} environment use a learning rate of $\alpha_{\bm{\theta}}=0.01$ for the variational and $\alpha_{\bm{\lambda}}=0.1$ for the state scaling parameters. In all other environments, a value of $\alpha=0.1$ is used for all parameter sets. A similar distinction is made w.r.t.\ parameter initialization, where \texttt{CartPole-v0} setups select $\theta \sim \mathcal{N}\left( 0, 0.1 \right)$, while the base option is always to draw the variational parameters uniformly at random from $\left( - \pi, \pi \right]$. The state scaling parameters are all initialized to the constant value $1.0$. The parameter update is performed using the \emph{Adam} optimizer~\cite{Kingma_2014}, modified with the \emph{AMSGrad} adjustment~\cite{Reddi_2019}. A discount factor of $\gamma=0.99$ is used in all cases. No baseline function is used in any of the environments, as performance was found to be sophisticated even without. If not stated differently, the architecture from \cref{fig:vqc} with a depth of $d=1$ is used, where the number of qubits is adjusted to match the state dimensionality. In order to make \gls{rl} training curves a bit more stable, the results are usually averaged over ten independent runs. Additionally, the performance is averaged over the last $20$ episodes (displayed in darker colors). Some plots also denote the performance of a random agent with a black dashed line and the optimal expected reward with a solid black one.

To support the results from \cref{sec:experiments}, we also conducted experiments for other setups and environments. Basically, the qualitative observations were always consistent with the claims we made, although the peculiarity was sometimes weaker or stronger.

\subsection{\label{subsec:deeper_circuits}Increased Quantum Circuit Depth}

\input{figures/post-processing_deeper}

Instead of using circuits with depth $d=1$, we use data re-uploading with depth $d=2$ on the \texttt{CartPole-v0} environment. Due to this, and as the resulting circuits contain $40$ instead of $24$ parameters, the \gls{rl} performance intuitively should improve. In fact, that is what can be observed in \cref{fig:partitioning_deep}. Compared to the results in \cref{fig:partitioning}, the convergence speed of the policies $\pi^{G_{f_{\mathcal{C}}}}_{=3}$ and $\pi^{G_{f_{\mathcal{C}}}}_{=1}$ has improved quite a bit. Also, the least-global policy is finally able to learn a close to optimal policy after $1000$ additional episodes, which was not the case previously. Initially, it actually outperforms the policy with the higher policy but caught up with after approx. $600$ episodes. We consider this behavior to be caused by statistical fluctuation, which is common for \gls{rl} training. As the global policy $\pi^{G_{f_{\mathcal{C}}}}_{=4}$ showed already a good performance for $d=1$, there was not much room for improvement. Again we back up the results by the respective effective dimensions in \cref{fig:partitioning_deep}. Here also the predicted pattern holds, with a slight overall improvement over the smaller models.

Overall we conclude, that the increasing model complexity benefits the \gls{rl} performance and general expressibility of the model for all policy formulations. Still, there is a strong correlation between globality and \gls{rl} performance, although it is slightly less pronounced than in the original setting. As it is highly desirable for \gls{nisq} hardware to keep the circuit complexity as low as possible, using a sophisticated post-processing function should be preferred over increasing the circuit depth.

\subsection{\label{subsec:other_environments}Extension to Other Environments}

\input{figures/post-processing_otherenvs}

\input{tables/fisher_spectrum}

To really make use of the construction of global post-processing functions for larger action spaces proposed by \cref{lem:optimal_partitioning}, we now take a look at two additional environments. The first one is the gridworld environment \texttt{FrozenLake}~\cite{Brockman_2016}, which has to decide between $4$ possible actions in every step. Consequently, the lowest possible globality value of a suitable policy is $G_{f_{\mathcal{C}}} = \log_2(4) = 2$, while an optimal formulation satisfies $G_{f_{\mathcal{C}}} = n = 4$. The training results for those two policies and different circuit depth is depicted at the top of \cref{fig:partitioning_additional}. We can basically observe the same pattern as throughout this paper, where a more global post-processing function improves the convergence speed. Also increasing the overall model complexity benefits the performance, wherefore the gap in performance for the different policy formulation decreases.

The second choice is a \texttt{ContextualBandits} environment with $32$ states and $8$ actions. As we encode the states via $1$-qubit rotations in a binary fashion, the VQC has $\log_2(32) = 5$ qubits. This implies an upper bound of $G_{f_{\mathcal{C}}} \leq 5$ and a lower bound of $G_{f_{\mathcal{C}}} \geq 3$. The performance of the two models is depicted in the lower part of \cref{fig:partitioning_additional}. Also here the predicted correlation can be observed, although the difference is not that significant. This might be partly down to the reason, that both models struggle to come close to the optimal expected reward of $1.0$. By using more sophisticated encoding schemes, or bigger models, one should be able to change this.

We also computed the associated effective dimensions and Fisher spectra, which again followed the predicted scheme. Overall it can be concluded, that the proportionality of the globality measure associated with a post-processing function and its \gls{rl} performance translates to a variety of environments.


\section{\label{sec:additional_analysis_fisher_spectrum}Abstracted Analysis of Fisher Spectrum}

As the analysis of the Fisher information spectrum is a powerful tool to assess trainability, we apply it to a range of different setups. We keep things as general as possible by sampling state values uniformly at random from $\left( -\pi, \pi \right]^{\abs{\bm{s}}}$. Due to the periodicity of the rotation gates used for encoding, this should cover a wide range of potential scenarios. The most critical property of the Fisher information spectrum is the concentration of eigenvalues close to $0$. A high proportion of small eigenvalues indicates a flat parameter space, which makes optimization with any gradient-based technique difficult.

The results for systems ranging from $4$ to $10$ qubits and depths $d=1$ to $d=4$ are summarized in \cref{tab:fisher_eigenvalues}. All experiments assumed an action space of size $2$. Most interestingly, the percentage of small eigenvalues for a global post-processing function is almost negligible in all cases. On the contrary, the spectra for policies based on post-processing functions with $G_{\mathcal{C}}=1$ are quite degenerated. As one would expect, the post-processing functions with $G_{\mathcal{C}}=3$ start out quite well, yet the farther they deviate from the optimal globality value, the more degeneration occurs. These results solidify the statement, that models with global post-processing functions benefit a wide range of applications.

The convergence of eigenvalues towards $0$ does not seem to be proportional to the system size. This has some potential implications w.r.t.\ the barren plateau problem, which is closely related to the trainability of a model. The term describes the observation, that the expectation value and also the variance of the gradients w.r.t.\ the parameters decrease exponentially with the number of qubits~\cite{Mcclean_2018}. \citeauthor{Abbas_2021} relates this to the Fisher information spectrum, i.e.\ the model is vulnerable to barren plateaus, iff $\mathrm{tr} \left( \mathbb{E}_{\Theta} \left[ F(\Theta) \right] \right)$ decreases exponentially with increasing system size~\cite{Mcclean_2018}. Following \cref{tab:fisher_eigenvalues}, no setup shows a progressive convergence of eigenvalues to $0$, although there are some quantitative differences. The considered model sizes are probably still too small for barren plateaus to occur, so for concluding statements additional investigation is necessary. However, following the above statement, barren plateaus are at least unlikely to occur, especially for small-scale models with global post-processing functions. Similar observations have been made in other fields of \gls{qml}~\cite{Abbas_2021,Kashif_2023}. However, the interpretation of the global post-processing function as global measurement potentially makes barren plateaus inevitable for increasing system size as shown in \citeauthor{Cerezo_2021} -- although the validity for large action spaces is not immediate. For a larger qubit count and circuit depth the proposed post-processing technique also allows adjusting the globality. This can be used to find a good balance between the empirical performance improvement demonstrated in this paper and reduced globality to prevent barren plateaus.

Last but not least, it has to be stated, that a problem related to barren plateaus is also known in the classical case. More concretely, big deep neural networks often suffer from vanishing gradients~\cite{Hochreiter_1998}. If the results presented in this section can be extended to larger quantum systems, the improvement in terms of trainability compared to classical models might point towards a possible quantum advantage.

%% file: figures/restricted_softmax.tex
\begin{figure}
    \centering
    \begin{tikzpicture}
        \centering
        \begin{axis}[
        	xlabel=episode,
        	ylabel=average collected reward,
        	title={Training Performance},
        	grid=both,
        	xmin=-30,xmax=1530,
            ymin=0.2,ymax=0.99,
            tick label style={font=\footnotesize},
        	width=0.9\linewidth,
        	height=.25\textheight,
        	axis x line=bottom, axis y line=left, tick align = outside,
            legend columns=-1,
        	legend style={/tikz/every even column/.append style={column sep=0.1cm},at={(0.47,1)},anchor=south,yshift=-4.5mm,font=\footnotesize},
        	no marks]
        	  \addplot[line width=.5pt,solid,color=darkgreen!40] %
            	table[x=episode,y=raw-d2,col sep=comma]{figures/data/data_restricted_softmax_reduced.csv};
            \addplot[line width=.5pt,solid,color=blue!40] %
            	table[x=episode,y=rsm-d2,col sep=comma]{figures/data/data_restricted_softmax_reduced.csv};
            \addplot[line width=.5pt,solid,color=violet!40] %
            	table[x=episode,y=rsm-d4,col sep=comma]{figures/data/data_restricted_softmax_reduced.csv};
            \addplot[line width=.5pt,solid,color=orange!40] %
            	table[x=episode,y=sm-d4,col sep=comma]{figures/data/data_restricted_softmax_reduced.csv};
            \addplot[line width=1pt,solid,color=darkgreen] %
            	table[x=episode,y=raw-d2-avg,col sep=comma]{figures/data/data_restricted_softmax_reduced.csv};
            \addplot[line width=1pt,solid,color=blue] %
            	table[x=episode,y=rsm-d2-avg,col sep=comma]{figures/data/data_restricted_softmax_reduced.csv};
            \addplot[line width=1pt,solid,color=violet] %
            	table[x=episode,y=rsm-d4-avg,col sep=comma]{figures/data/data_restricted_softmax_reduced.csv};
            \addplot[line width=1pt,solid,color=orange] %
            	table[x=episode,y=sm-d4-avg,col sep=comma]{figures/data/data_restricted_softmax_reduced.csv};
            \addplot[color=red, domain=-25:1524, line width=1pt] {0.75};
            \addplot[dashed, color=black, domain=-25:1524, line width=1.2pt] {0.25};
            \legend{,,,,$\pi^{raw}_{d=2}$,$\pi^{r-sm}_{d=2}$,$\pi^{r-sm}_{d=4}$,$\pi^{sm}_{d=4}$}
        \end{axis}
    \end{tikzpicture}
    \caption{\label{fig:restricted_softmax}Experiments on a \texttt{ContextualBandits} environment (which is uniform following \cref{def:uniform_env}) demonstrate the implications of \cref{lem:restricted_softmax}. The reward structure is defined in a way, such that the expected reward is equivalent to the introduced notion of accuracy $ACC_{\pi}(\mathcal{E}_{\mathcal{A}})$. A \texttt{RESTRICTED}-\texttt{SOFTMAX}-\gls{vqc} $\pi^{r-sm}$ (with circuit depth $d=2$ and $d=4$) does not perform even close to the theoretical limit of $0.75$, indicated by the red horizontal line, while the \texttt{RAW}-\gls{vqc} $\pi^{raw}$ and \texttt{SOFTMAX}-\gls{vqc} $\pi^{sm}$ surpass this value. The \texttt{RAW}-\gls{vqc} outperforms the \texttt{SOFTMAX}-\gls{vqc}, even with a shallower circuit. To achieve this, we used the considerations on global policy construction from \cref{sec:theory}. All results are averaged over $50$ independent experiments.}
\end{figure}

%% file: tables/extracted_info.tex
\begin{table}
    \centering
        \caption{\label{tab:extracted_information}Extracted information of the post-processing function $f_{\mathcal{C}}$ for the partitioning $\mathcal{C}$ given in \cref{eq:C_example_1,eq:C_example_2,eq:C_example_3,eq:C_example_4} for all $16$ bitstrings. The marked bits are used to get an unambiguous assignment to the respective partitions. It is easy to check that one can not go with less information, consequently, the count corresponds to the extracted information.\\}
        \begin{tabular}{ccccc||cccc}    
            \toprule
            &bitstring $\bm{b}$ & $EI_{f_{\mathcal{C}}}(\bm{b})$ & containing partition && bitstring $\bm{b}$ & $EI_{f_{\mathcal{C}}}(\bm{b})$ & containing partition & \\
            \midrule
            &\fcolorbox{black}{white}{0}\fcolorbox{white}{white}{0}\fcolorbox{white}{white}{0}\fcolorbox{black}{white}{0} & 2 & $\mathcal{C}_{a=0}$ && \fcolorbox{black}{white}{1}\fcolorbox{black}{white}{0}\fcolorbox{white}{white}{0}\fcolorbox{black}{white}{0} & 3 & $\mathcal{C}_{a=2}$ &\\
            &\fcolorbox{black}{white}{0}\fcolorbox{white}{white}{0}\fcolorbox{white}{white}{0}\fcolorbox{black}{white}{1} & 2 & $\mathcal{C}_{a=1}$ && \fcolorbox{black}{white}{1}\fcolorbox{black}{white}{0}\fcolorbox{white}{white}{0}\fcolorbox{black}{white}{1} & 3 & $\mathcal{C}_{a=3}$ &\\
            &\fcolorbox{black}{white}{0}\fcolorbox{white}{white}{0}\fcolorbox{white}{white}{1}\fcolorbox{black}{white}{0} & 2 & $\mathcal{C}_{a=0}$ && \fcolorbox{black}{white}{1}\fcolorbox{black}{white}{0}\fcolorbox{white}{white}{1}\fcolorbox{black}{white}{0} & 3 & $\mathcal{C}_{a=2}$ &\\
            &\fcolorbox{black}{white}{0}\fcolorbox{white}{white}{0}\fcolorbox{white}{white}{1}\fcolorbox{black}{white}{1} & 2 & $\mathcal{C}_{a=1}$ && \fcolorbox{black}{white}{1}\fcolorbox{black}{white}{0}\fcolorbox{white}{white}{1}\fcolorbox{black}{white}{1} & 3 & $\mathcal{C}_{a=3}$ &\\
            &\fcolorbox{black}{white}{0}\fcolorbox{white}{white}{1}\fcolorbox{white}{white}{0}\fcolorbox{black}{white}{0} & 2 & $\mathcal{C}_{a=0}$ && \fcolorbox{black}{white}{1}\fcolorbox{black}{white}{1}\fcolorbox{white}{white}{0}\fcolorbox{black}{white}{0} & 3 & $\mathcal{C}_{a=3}$ &\\
            &\fcolorbox{black}{white}{0}\fcolorbox{white}{white}{1}\fcolorbox{white}{white}{0}\fcolorbox{black}{white}{1} & 2 & $\mathcal{C}_{a=1}$ && \fcolorbox{black}{white}{1}\fcolorbox{black}{white}{1}\fcolorbox{white}{white}{0}\fcolorbox{black}{white}{1} & 3 & $\mathcal{C}_{a=2}$ &\\
            &\fcolorbox{black}{white}{0}\fcolorbox{white}{white}{1}\fcolorbox{white}{white}{1}\fcolorbox{black}{white}{0} & 2 & $\mathcal{C}_{a=0}$ && \fcolorbox{black}{white}{1}\fcolorbox{black}{white}{1}\fcolorbox{white}{white}{1}\fcolorbox{black}{white}{0} & 3 & $\mathcal{C}_{a=3}$ &\\
            &\fcolorbox{black}{white}{0}\fcolorbox{white}{white}{1}\fcolorbox{white}{white}{1}\fcolorbox{black}{white}{1} & 2 & $\mathcal{C}_{a=1}$ && \fcolorbox{black}{white}{1}\fcolorbox{black}{white}{1}\fcolorbox{white}{white}{1}\fcolorbox{black}{white}{1} & 3 & $\mathcal{C}_{a=2}$ &\\
            \bottomrule
        \end{tabular}
\end{table}

%% file: figures/partitioning_histogram.tex
\begin{figure}
    \centering
    \subfigure[Histogram for $n=4$ and $\abs{\mathcal{A}}=2$. Of all $6435$ possible partitionings only one instance exhibits the optimal globality value $G_{f_\mathcal{C}} = 4$.]{
    \begin{tikzpicture}
    \begin{axis}[name=plot1,
        width=0.45\linewidth, height=.15\textheight,
        grid=major,
        tick label style={font=\footnotesize},
        minor y tick  num=1,
        xmin=0.9,xmax=4.15,
        ymin=0,ymax=0.27,
        ylabel={normalized count},
        xlabel={$G_{f_\mathcal{C}}$},
        xtick={1,1.5,2,2.5,3,3.5,4},
        xticklabels={$1.0$,$1.5$,$2.0$,$2.5$,$3.0$,$3.5$,$4.0$},
        ybar = .05cm,
        bar width = 6pt,
        axis x line=bottom, axis y line=left, tick align = outside,
        legend columns=-1,
        ]
        \addplot[fill=blue!66,ybar,no marks,error bars/.cd, y dir=both, y explicit] coordinates {
            (1.0,0.00064)
            (2.0,0.02506)
            (2.125,0.03427)
            (2.25,0.10125)
            (2.375,0.03011)
            (2.5,0.25337)
            (2.625,0.03738)
            (2.75,0.14174)
            (2.875,0.10436)
            (3.0,0.11007)
            (3.125,0.11422)
            (3.25,0.03583)
            (3.375,0.00519)
            (3.5,0.00519)
            (3.625,0.0)
            (3.75,0.0)
            (3.875,0.0)
            (4.0,0.00016)
        };
    \end{axis}
    \end{tikzpicture}
    }
    \qquad
    \subfigure[Histogram for $n=6$ and $\abs{\mathcal{A}}=4$. It is basically impossible to guess an partitioning with optimal globality value $G_{f_\mathcal{C}} = 6$ from all $2.8 \cdot 10^{34}$ possibilities.]{
    \begin{tikzpicture}
    \begin{axis}[name=plot2,
        width=0.45\linewidth, height=.15\textheight,
        grid=major,
        tick label style={font=\footnotesize},
        minor y tick  num=1,
        xmin=1.9,xmax=6.15,
        ymin=0,ymax=0.39,
        ylabel={normalized count},
        xlabel={$G_{f_\mathcal{C}}$},
        xtick={2,3,4,5,6},
        xticklabels={$2.0$,$3.0$,$4.0$,$5.0$,$6.0$},
        ybar = .05cm,
        bar width = 4.5pt,
        axis x line=bottom, axis y line=left, tick align = outside,
        legend columns=-1,
        ]
        \addplot[fill=blue!66,ybar,no marks,error bars/.cd, y dir=both, y explicit] coordinates {
            (3.25,0.01274)
            (3.375,0.03264)
            (3.5,0.04618)
            (3.625,0.25717)
            (3.75,0.38376)
            (3.875,0.22293)
            (4.0,0.04459)
        };
    \end{axis}
    \end{tikzpicture}
    }
    \caption{\label{fig:globality_dist}Histogram of globality values over all possible partitionings.}
    \vspace{-2mm}
\end{figure}
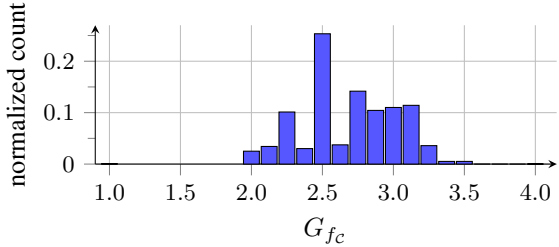
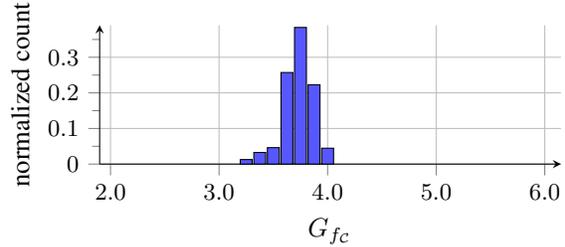

%% file: figures/post-processing_deeper.tex
\begin{figure}
    \centering
    \subfigure[RL training performance]{
    \begin{tikzpicture}
        \centering
        \begin{axis}[
            name=plot1,
        	xlabel=episode,
        	ylabel=average collected reward,
        	grid=both,
        	xmin=-30,xmax=1030,
            ymin=0,ymax=210,
            tick label style={font=\footnotesize},
            xtick={0,200,400,600,800,1000},
        	width=0.45\linewidth,
        	height=.25\textheight,
        	axis x line=bottom, axis y line=left, tick align = outside,
            legend columns=1,
        	legend style={/tikz/every even column/.append style={column sep=0.1cm, row sep=0.1cm},at={(1.0,0.38)},anchor=east,yshift=-5mm,font=\footnotesize},
        	no marks]
        	\addplot[line width=.5pt,solid,color=darkgreen!40] %
            	table[x=episode,y=glob-avg,col sep=comma]{figures/data/data_partitioning_deep.csv};
            \addplot[line width=.5pt,solid,color=red!40] %
            	table[x=episode,y=loc3-avg,col sep=comma]{figures/data/data_partitioning_deep.csv};
            \addplot[line width=.5pt,solid,color=violet!40] %
            	table[x=episode,y=loc1-avg,col sep=comma]{figures/data/data_partitioning_deep.csv};
            \addplot[line width=1pt,solid,color=darkgreen] %
            	table[x=episode,y=glob,col sep=comma]{figures/data/data_partitioning_deep.csv};
            \addplot[line width=1pt,solid,color=red] %
            	table[x=episode,y=loc3,col sep=comma]{figures/data/data_partitioning_deep.csv};
            \addplot[line width=1pt,solid,color=violet] %
            	table[x=episode,y=loc1,col sep=comma]{figures/data/data_partitioning_deep.csv};
            \legend{,,,$\pi^{G_{f_{\mathcal{C}}}}_{=4}~(d=2)$,$\pi^{G_{f_{\mathcal{C}}}}_{=3}~(d=2)$,$\pi^{G_{f_{\mathcal{C}}}}_{=1}~(d=2)$}
            \addplot[color=black, domain=-25:1024, line width=1pt] {200.0};
            \addplot[color=white, domain=-10:1024, line width=5.5pt] {206.0};
            \addplot[dashed, color=black, domain=-25:1024, line width=1pt] {23.5};
        \end{axis}
        \end{tikzpicture}
        }
        \qquad
        \subfigure[Normalized Effective dimension]{
        \begin{tikzpicture}
        \begin{axis}[
            name=plot2,
        	xlabel=data size,
        	ylabel=normalized effective dimension,
        	grid=both,
        	xmin=-50,xmax=1000050,
            ymin=0.25,ymax=0.875,
            tick label style={font=\footnotesize},
        	width=0.45\linewidth,
        	height=.25\textheight,
        	xtick={5000,8000,10000,40000,60000,100000,150000,200000,500000,1000000},
            xticklabels={,,,,,0.1,,,0.5,1.0},
        	axis x line=bottom, axis y line=left, tick align = outside,
            legend columns=-1,
        	legend style={/tikz/every even column/.append style={column sep=0.1cm},at={(0.5,1.0)},anchor=south,yshift=-5mm,font=\footnotesize},
        	]
            \addplot[line width=1pt,solid,color=darkgreen,mark=+] %
            	table[x=data-size,y=g_4,col sep=comma]{figures/data/data_effdim_deep.csv};
            \addplot[line width=1pt,solid,color=red,mark=+] %
            	table[x=data-size,y=g_3,col sep=comma]{figures/data/data_effdim_deep.csv};
            \addplot[line width=1pt,solid,color=violet,mark=+] %
            	table[x=data-size,y=g_1,col sep=comma]{figures/data/data_effdim_deep.csv};
        \end{axis}
    \end{tikzpicture}
    }
    \caption{\label{fig:partitioning_deep}\gls{rl} training performance and associated effective dimension on the \texttt{CartPole-v0} environment. The same setup and policy formulations as in \cref{fig:partitioning} are used, only the depth of the underlying \glspl{vqc} is increased to $d=2$.}
\end{figure}
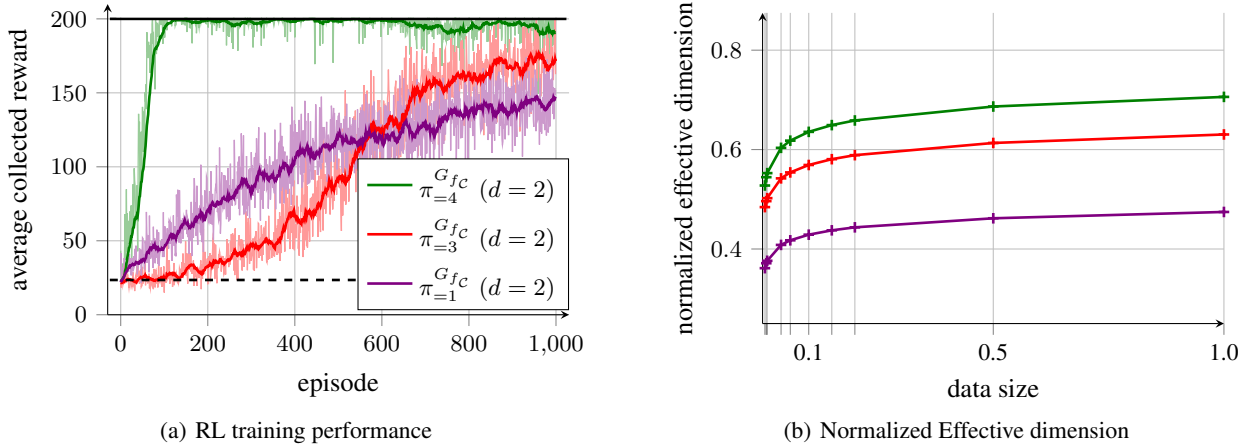

%% file: figures/post-processing_otherenvs.tex
\begin{figure}
    \centering
    \subfigure[\texttt{FrozenLake} environment with $16$ discrete states and $4$ actions.]{
    \begin{tikzpicture}
        \centering
        \begin{axis}[
            name=plot1,
        	xlabel=episode,
        	ylabel=average collected reward,
        	grid=both,
        	xmin=-30,xmax=2530,
            ymin=-30,ymax=80,
            tick label style={font=\footnotesize},
        	width=0.43\linewidth,
        	height=.25\textheight,
        	axis x line=bottom, axis y line=left, tick align = outside,
            legend columns=1,
        	legend style={/tikz/every even column/.append style={column sep=0.1cm, row sep=0.1cm},at={(1.0,0.44)},anchor=east,yshift=-5mm,font=\footnotesize},
        	no marks]
        	\addplot[line width=.5pt,solid,color=darkgreen!40] %
            	table[x=episode,y=glob-d2-avg,col sep=comma]{figures/data/data_partitioning_frozen_reduced.csv};
        	\addplot[line width=.5pt,solid,color=teal!40] %
            	table[x=episode,y=glob-d1-avg,col sep=comma]{figures/data/data_partitioning_frozen_reduced.csv};
            \addplot[line width=.5pt,solid,color=violet!40] %
            	table[x=episode,y=loc2-d2-avg,col sep=comma]{figures/data/data_partitioning_frozen_reduced.csv};
            \addplot[line width=.5pt,solid,color=red!40] %
            	table[x=episode,y=loc2-d1-avg,col sep=comma]{figures/data/data_partitioning_frozen_reduced.csv};
            \addplot[line width=1pt,solid,color=darkgreen] %
            	table[x=episode,y=glob-d2,col sep=comma]{figures/data/data_partitioning_frozen_reduced.csv};
            \addplot[line width=1pt,solid,color=teal] %
            	table[x=episode,y=glob-d1,col sep=comma]{figures/data/data_partitioning_frozen_reduced.csv};
            \addplot[line width=1pt,solid,color=violet] %
            	table[x=episode,y=loc2-d2,col sep=comma]{figures/data/data_partitioning_frozen_reduced.csv};
            \addplot[line width=1pt,solid,color=red] %
            	table[x=episode,y=loc2-d1,col sep=comma]{figures/data/data_partitioning_frozen_reduced.csv};
            \addplot[dashed, color=black, domain=-25:2524, line width=1pt] {-26.0};
            \legend{,,,,$\pi^{G_{f_{\mathcal{C}}}}_{=4}~(d=2)$,$\pi^{G_{f_{\mathcal{C}}}}_{=4}~(d=1)$,$\pi^{G_{f_{\mathcal{C}}}}_{=2}~(d=2)$,$\pi^{G_{f_{\mathcal{C}}}}_{=2}~(d=1)$}
        \end{axis}
        \end{tikzpicture}
        }
        \qquad
        \subfigure[\texttt{ContextualBandits} environment with $32$ discrete states and $8$ actions.]{
        \begin{tikzpicture}
        \begin{axis}[
            name=plot2,
        	xlabel=episode,
        	ylabel=average collected reward,
        	grid=both,
        	xmin=-30,xmax=10030,
            ymin=-0.8,ymax=0.4,
            tick label style={font=\footnotesize},
        	width=0.43\linewidth,
        	height=.25\textheight,
        	axis x line=bottom, axis y line=left, tick align = outside,
            legend columns=1,
        	legend style={/tikz/every even column/.append style={column sep=0.1cm, row sep=0.1cm},at={(1.0,0.29)},anchor=east,yshift=-5mm,font=\footnotesize},
        	no marks]
            \addplot[line width=.5pt,solid,color=blue!40] %
            	table[x=episode,y=glob-avg,col sep=comma]{figures/data/data_partitioning_bandit_reduced.csv};
            \addplot[line width=.5pt,solid,color=orange!40] %
            	table[x=episode,y=loc3-avg,col sep=comma]{figures/data/data_partitioning_bandit_reduced.csv};
            \addplot[line width=1pt,solid,color=blue]
            	table[x=episode,y=glob,col sep=comma]{figures/data/data_partitioning_bandit_reduced.csv};
            \addplot[line width=1pt,solid,color=orange] %
            	table[x=episode,y=loc3,col sep=comma]{figures/data/data_partitioning_bandit_reduced.csv};
            \addplot[dashed, color=black, domain=-25:10024, line width=1pt] {-0.75};
            \legend{,,$\pi^{G_{f_{\mathcal{C}}}}_{=5}~(d=2)$,$\pi^{G_{f_{\mathcal{C}}}}_{=3}~(d=2)$}
        \end{axis}
    \end{tikzpicture}
    }
    \caption{\label{fig:partitioning_additional}\gls{rl} training performance for different environments. The underlying global post-processing functions are determined as described in \cref{subsubsec:optimal_partitioning}. These results are averaged over $100$ independent experiments;}
\end{figure}
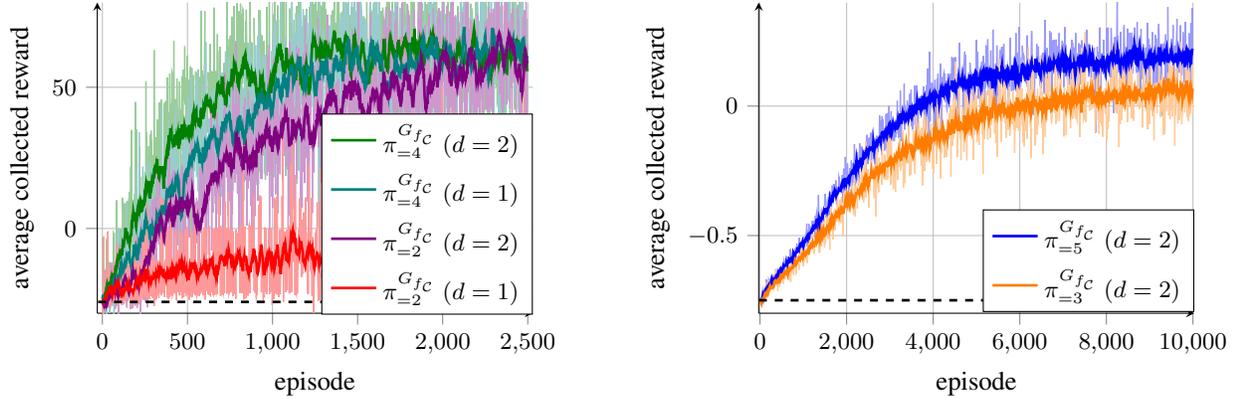

%% file: tables/fisher_spectrum.tex
\begin{table*}
    \centering
        \caption{\label{tab:fisher_eigenvalues}Percentage of eigenvalues of \gls{fim} for two actions and models of increasing complexity that are close to zero ($<10^{-7}$). The empirical \gls{fim} is estimated with $100$ random parameter sets for each of the $100$ random states $\bm{s}$. Unlike in \cref{fig:effdim_spectrum}, the states elements are sampled uniformly at random from $\left[ -\pi, \pi \right)$ to allow for statements abstracted from a concrete \gls{rl} environment.\\}
        \begin{tabular}{cc|cccc}
            \toprule
            & & depth $d=1$ & depth $d=2$ & depth $d=3$ & depth $d=4$ \\
            \midrule
            \midrule
            & $\pi_{G_{f_{\mathcal{C}}}=4}$ & $0\%$ & & & \\
            $4$ qubits &$\pi_{G_{f_{\mathcal{C}}}=3}$ & $17\%$ & & & \\
            ($\abs{\Theta} = 24,-,-,-$) &$\pi_{G_{f_{\mathcal{C}}}=1}$ & $50\%$ & & &\\
            \midrule
            & $\pi_{G_{f_{\mathcal{C}}}=6}$ & $0\%$ & $0\%$ & & \\
            $6$ qubits &$\pi_{G_{f_{\mathcal{C}}}=3}$ & $26\%$ & $20\%$ & & \\
            ($\abs{\Theta} = 36,60,-,-$) &$\pi_{G_{f_{\mathcal{C}}}=1}$ & $48\%$ & $33\%$ & &\\
            \midrule
            & $\pi_{G_{f_{\mathcal{C}}}=8}$ & $0\%$ & $0\%$ & $0\%$ & \\
            $8$ qubits &$\pi_{G_{f_{\mathcal{C}}}=3}$ & $30\%$ & $21\%$ & $18\%$ & \\
            ($\abs{\Theta} = 48,80,112,-$) &$\pi_{G_{f_{\mathcal{C}}}=1}$ & $38\%$ & $26\%$ & $25\%$ &\\
            \midrule
            & $\pi_{G_{f_{\mathcal{C}}}=10}$ & $1\%$ & $0\%$ & $0\%$ & $0\%$ \\
            $10$ qubits &$\pi_{G_{f_{\mathcal{C}}}=3}$ & $28\%$ & $27\%$ & $19\%$ & $16\%$ \\
            ($\abs{\Theta} = 60,100,140,180$) &$\pi_{G_{f_{\mathcal{C}}}=1}$ & $30\%$ & $28\%$ & $23\%$ & $20\%$ \\
            \bottomrule
        \end{tabular}
\end{table*}

